\documentclass{article}
\usepackage[utf8]{inputenc}
\usepackage{enumitem}
\usepackage{amsthm}
\usepackage{amsmath}
\usepackage{amssymb}
\usepackage{upquote}
\usepackage[margin=1.4in]{geometry}
\usepackage[backend=biber,style=numeric,sortcites,natbib=true,sorting=none]{biblatex}
\usepackage{mathtools}
\usepackage{physics}
\usepackage{qcircuit}

\usepackage{thmtools}
\usepackage{thm-restate}
\usepackage{hyperref}
\usepackage{cleveref}

\newtheorem{theorem}{Theorem}[section]
\newtheorem{definition}[theorem]{Definition}
\newtheorem{lemma}[theorem]{Lemma}
\newtheorem{corollary}[theorem]{Corollary}

\newcommand{\Z}{\mathbb{Z}}

\newcommand{\poly}{\textmd{poly}}
\newcommand{\qwelipse}{\qw \raisebox{.3em}{\ldots}}

\title{Unitarization Through Approximate Basis}
\author{ Joshua Cook } 

\begin{document}

\maketitle

\begin{abstract}
    We introduce the problem of unitarization. Unitarization is the problem of taking $k$ input quantum circuits that produce orthogonal states from the all $0$ state, and create an output circuit implementing a unitary with its first $k$ columns as those states. That is, the output circuit takes the $k$th computational basis state to the state prepared by the $k$th input circuit. We allow the output circuit to use ancilla qubits initialized to $0$. But ancilla qubits must always be returned to $0$ for \emph{any} input. Further, the input circuits may also use ancilla qubits, but we are only guaranteed the input circuits returns ancilla qubits to $0$ on the all $0$ input. Finally, we want it to leave any orthogonal states unchanged.
    
    The unitarization problem seems hard if the output states are neither orthogonal to or in the span of the computational basis states that need to map to them. In this work, we approximately solve this problem in the case where input circuits are given as black box oracles by probably finding an approximate basis for our states. This method may be more interesting than the application. This technique can be thought of as a sort of quantum analogue of Gram-Schmidt orthogonalization for quantum states.
        
    Specifically, we find an approximate basis in polynomial time for the following parameters. Take any natural $n$, $k = O\left(\frac{\ln(n)}{\ln(\ln(n))}\right)$, and $\epsilon = 2^{-O(\sqrt{\ln(n)})}$. Take any $k$ input quantum states, $(\ket{\psi_i})_{i\in [k]}$, on polynomial in $n$ qubits prepared by quantum oracles, $(V_i)_{i \in [k]}$ (that we can control call and control invert). Then there is a quantum circuit with polynomial size in $n$ with access to the oracles $(V_i)_{i \in [k]}$ that with at least $1 - \epsilon$ probability, computes at most $k$ circuits with size polynomial in $n$ and oracle access to $(V_i)_{i \in [k]}$ that $\epsilon$ approximately computes an $\epsilon$ approximate orthonormal basis for $(\ket{\psi_i})_{i\in [k]}$.
\end{abstract}

\section{Introduction}

Quantum computing is a way of performing computation that can simulate classical computation, but works in a fundamentally different way. Intuitively, quantum computers can achieve better performance than classical computers by using quantum interference. The requirements for this interference to happen restricts what state can be maintained in a way classical computers don't.

This often manifests in what is lovingly referred to as ``entangled garbage''. Often a quantum algorithm wants a circuit that for some function \(f:[N]\rightarrow [N]\) takes \(\ket{i}\) to \(\ket{f(i)}\). Sometimes though, we only have an algorithm that will take \(\ket{i}\ket{0^m}\) to \(\ket{f(i)}\ket{g(i)}\) for some function \(g\). If the function \(g\) is not constant, then it in some way ``records'' the input $i$ as entangled garbage. In classical computation, this distinction does not matter, since the extra bits can be ignored or reset. But for quantum algorithms, this entangled garbage can prevent interference.

We want to be able to perform this sort of mapping for an $f$ that maps $[N]$ to quantum states, where the mappings are given explicitly by a circuit for each $i \in [N]$ producing the state that $i$ should map to. In particular, we want to implement $f$ in such a way that there is no garbage on any input, and that we effect as few states as possible.

While in principle any $f$ mapping any inputs to orthogonal output states can be implemented by a unitary that does not use the ancilla qubits at all, using ancilla qubits may be very useful. In this work, we want to provide tools for efficiently constructing this type of a mapping for as general of an $f$ as possible without leaving any entangled garbage. In particular, we want to never produce garbage on \emph{any} input, and only disturb the input on the smallest subspace possible.

We show that if $f$ is some sparse function, that only takes a few ($O(\frac{\log(n)}{\log(\log(n))})$) orthogonal inputs states to a few orthogonal outputs, and we have a few circuits that can prepare these orthogonal states, then we can implement an $O(2^{O(-\sqrt{\log(n)})})$ approximation of $f$ in polynomial time. This is done through an iterative orthogonalization process for each state, and then careful order of this orthogonalization to keep overhead low.

\subsection{Result Overview}

In all our results, we will suppress a factor of $n$ for the number of qubits our state is on. This will simplify the analysis since most gates are controlled nots of the state being in a specific computational basis state, swaps of computational basis states, or application of an oracle. We also say a unitary $U$ on $n + p$ qubits prepares state $\ket{\psi}$ on $n$ qubits if $U\ket{0^n}\ket{0^p} = \ket{\psi}\ket{0^p}$.

The uniterization problem we try to solve is:
\begin{restatable}[Unitarization Definition]{definition}{UnitarizationDefinition}
\label{Unitarization}
    For natural $k$ and $\epsilon >0$, the $k$, $\epsilon$ unitierization problem is to find an oracle circuit $C$ so that:

    Given \((\ket{\psi_i})_{i \in [k]}\) orthonormal states on $n$ qubits, where for each $i \in [k]$, $\ket{\psi_i}$ is prepared by circuit \(V_i\), there is a unitary $U$ on $n$ qubits so that:
    \begin{enumerate}
        \item For all \(i \in [k]\), \(U\ket{i-1} = \ket{\psi_i} \).
        \item For all states \(\ket{\theta}\) orthogonal to each \((\ket{\psi_i})_{i \in [k]}\) and \((\ket{i})_{i \in [k]}\), \(U\ket{\theta} = \ket{\theta}\).
        \item $C$ is a circuit with oracle access to \((V_i)_{i\in [k]}\) that \(\epsilon\)-approximates $U$.
    \end{enumerate}
\end{restatable}
This is the problem of, given a way to compute $k$ orthogonal states, $\epsilon$ approximate a unitary that maps the first $k$ computational states to these states. While this might not seem hard at first, we explain why it is difficult in \cref{WhyHard}.

Then our main theorem is that for a modest $k$, and $\epsilon$, we can solve the unitarization problem.
\begin{restatable}[Unitarization Solution]{theorem}{CleanUnitary}
\label{CleanUnitary}
    The $k$, $\epsilon$ unitarization problem can be solved with a circuit $C$ of size
    \[\poly\left( k^k \left(\frac{1}{\epsilon}\right)^{\min\{k, \ln(1/\epsilon)\}}\right),\]
    computable with probability $1 - \epsilon$ in time polynomial in its size.
\end{restatable}

Then by choosing appropriate $k$ and $\epsilon$, we can get polynomial time construction. For large $k$, we have
\begin{corollary}[Unitarization For Large $k$]
    For any \(k = O\left(\frac{\ln(n)}{\ln(\ln(n))}\right)\) and \(\epsilon = 2^{-O(\sqrt{\ln(n)})}\), the $k$, $\epsilon$ unitarization problem can be solved with a polynomial sized circuit,  computable with probability $1 - \epsilon$ in polynomial time.
\end{corollary}

For smaller $k$, especially $k < \sqrt{\log(n)}$, we have the simpler result:
\begin{corollary}[Unitarization For Small $k$]
    For any \(k = O\left(\frac{\ln(n)}{\ln(\ln(n))}\right)\) and \(\epsilon = 2^{-O(\log(n)/k)}\), the $k$, $\epsilon$ unitarization problem can be solved with a polynomial sized circuit, computable with probability $1 - \epsilon$ in polynomial time.
\end{corollary}
Notably, for $k = O(1)$, then $\epsilon = \poly(1/n)$.

Then \cref{CleanUnitary} is a corollary of our approximate basis theorem that the majority of this paper focuses on proving.
\begin{restatable}[Approximate Basis]{theorem}{ApproximateBasis}
\label{ApproximateBasis}
    For natural \(k\), let \((\ket{\psi_i})_{i \in [k]}\) be states on $n$ qubits prepared by \((V_i)_{i \in [k]}\).
    
    Then for every \(\epsilon > 0\), \(\gamma > 0\) there is an algorithm that runs in time 
    \[\poly\left(k^k \left(\frac{1}{\epsilon}\right)^{\min\{k, \ln(1/\epsilon)\}}\right) \ln(1/\gamma)\]
    that with probability at most \(\gamma\) fails to output some set of circuits $X$ such that:
    \begin{enumerate}
        \item For some orthonormal basis \(Y\), for all \(i \in [k]\), \(\Delta(Y, \ket{\psi_i}) < \epsilon\), and $X = (x_a)_{\ket{a} \in Y}$.
        \item $Y$ is in the subspace spanned by $(\ket{\psi_i})_{i \in [k]}$.
        \item For each $\ket{a} \in Y$, we have that \(x_a\) \(\eta_{x_a}\)-approximately prepares \(\ket{a}\).
        \item The total size of $X$ is small:
        \[\sum_{x \in X} |x| \leq \poly\left(k^k \left(\frac{1}{\epsilon}\right)^{\min\{k, \ln(1/\epsilon)\}}\right).\]
        \item The total error in the approximations is small:
        \[\sum_{x \in |X|} \eta_x \leq \epsilon .\]
    \end{enumerate}
\end{restatable}

The main tool used to prove \cref{ApproximateBasis} is an orthogonalization procedure. This procedure lets us take a circuit that only prepares one state \(\epsilon\) far from a computable subspace, and gives an approximation of a unitary that computes the orthogonal direction of that state. This combined with a Gram-Schmidt like process allows us to compute an approximate orthonormal basis.

First, to notate the orthogonal components:
\begin{definition}[Orthogonal Component]
    For some natural \(k\), let \((\ket{\phi_i}_{i \in [k]}\) be some orthonormal basis. Let \(\ket{\psi}\) be some state. Then for some unique state \(\ket{\phi}\),
    \[\ket{\psi} = \sum_{i \in [k]} (\bra{\phi_i}\ket{\psi}) \ket{\phi_i} + \left(\sqrt{1 - \sum_{i \in [k]} \|\bra{\phi_i}\ket{\psi}\|^2} \right)\ket{\phi}.\]
    Then define a function \(\Phi\) as
    \[\Phi((\ket{\phi_i}_{i \in [k]}, \ket{\psi}) = \ket{\phi}.\]
    Similarly, define function \(\Delta\) as
    \[\Delta((\ket{\phi_i}_{i \in [k]}, \ket{\psi}) = \sqrt{1 - \sum_{i \in [k]} \|\bra{\phi_i}\ket{\psi}\|^2} .\]
\end{definition}

Then our orthogonalization lemma is:
\begin{restatable}[State Orthogonalization]{lemma}{Orthogonalization}
\label{Orthogonalization}
    Let $\ket{\psi}$ be an $n$ qubit state prepared by $V$. Let $k$ be natural, and \((\ket{\phi_i})_{i \in [k]}\) be orthonormal states on $n$ qubits, where for each \(i \in [k]\), unitary \(U_i\) prepares \(\ket{\phi_i}\). For some \(\delta > 0\), suppose \(\Delta((\ket{\phi_i})_{i \in [k]}, \ket{\psi}) = \beta > \delta\).
    
    Then for all \(\epsilon > 0\), there exists a circuit, $C$, using oracle calls to $V$ and $U_i$ such that
    \begin{enumerate}
        \item $C$ \(\epsilon\)-approximately prepares $\Phi((\ket{\phi_i})_{i \in [k]}, \ket{\psi})$.
        \item $C$ has size \(O \left(\frac{k \ln(1/\epsilon)}{\delta^2}\right)\).
        \item $C$ can be constructed in time polynomial in its size.
        \item \(C\) only uses \(O \left(\frac{ \ln(1/\epsilon)}{\delta^2}\right)\) oracle calls to $V$, and each \(U_i\).
        \item Suppose $R$ is an oracle on $O(\log(k)) $ qubits so that \(R \ket{0} = \beta \ket{0} + \sum_{i \in [k]}  \bra{\phi_i}\ket{\psi} \ket{i} \).
        
        Then $C$ uses only \(O\left(\frac{\ln(1/\epsilon)}{\delta^2}\right)\) oracle calls to $R$.
    \end{enumerate}
\end{restatable}
The last requirement, about the $R$, is simply to separate the preprocessing for the orthogonal component circuit of $C$ from the circuit itself. This simplifies later analysis. An approximation of this $R$ can be computed in a straightforward manner.

Careful application of \cref{Orthogonalization} gives our approximate basis theorem, which we use to solve the unitarization problem..

\subsection{Why Isn't Unitarization Trivial?}
\label{WhyHard}

At first glance, unitarization may seem easy. It looks similar to a problem any new student in quantum computation might see. The concern is that we want our circuit $C$ to approximate $U$ on every input, which means that the ancilla qubits must always be returned to $0$.

To show the difficulty, let us first show some failed attempts and explain why they fail. Let us work in the simple case where we just have two states \(\ket{\psi_1}\) and \(\ket{\psi_2}\) prepared by an $n$ bit unitary \(V_1\) and \(V_2\) on \(\ket{1}\) and \(\ket{2}\).

First some notation. Informally (formally defined in \cref{Notation}), for basis state $a$, let $M_a$ be quantum operation that nots a qubit controlled on if that basis state is $a$. Similarly for any quantum state $\ket{\psi}$ that we can prepare, $M_\psi$ nots a qubit controlled on if the state is $\ket{\psi}$.

Now for a noble first attempt at unitarization for the case where we have 2 states. Let \(C_1\) be the following circuit:
$$
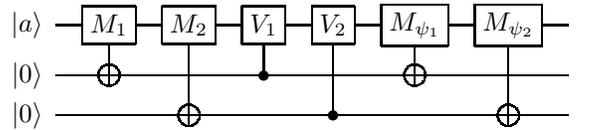
\Qcircuit @C=1em @R=.7em {
	\lstick{\ket{a}} & \gate{M_1} \qwx[1] & \gate{M_2} \qwx[2] & \gate{V_1} & \gate{V_2} &  \gate{M_{\psi_1}} \qwx[1] & \gate{M_{\psi_2}} \qwx[2] & \qw & \\
	\lstick{\ket{0}} & \targ & \qw & \ctrl{-1} & \qw & \targ & \qw & \qw & \\
	\lstick{\ket{0}} & \qw & \targ & \qw & \ctrl{-2} & \qw &\targ & \qw & \\
}.$$

We can see that for \(i \in [2]\), we indeed have \(C_1 \ket{i} \ket{0^2} = \ket{\psi_i} \ket{0^2}\). And for state, $\ket{\phi}$, orthogonal to \(\ket{1}, \ket{2}, \ket{\psi_1}, \) and \(\ket{\psi_2}\), $\ket{\phi}$ is unchanged by \(C_1\). But what about \(C_1 \ket{\psi_1}\ket{0^2}\)?  Well, if \(\ket{\psi_1}\) is in the span of \(\ket{1}\) and \(\ket{2}\), then we also have \(C_1 \ket{\psi_1} \ket{0^2}  = \ket{a}\ket{0^2}\) for some state \(\ket{a}\). But if \(\ket{\psi_1}\) is orthogonal to \(\ket{0}\) and \(\ket{1}\), then \(C_1\ket{\psi_1} = \ket{\psi_1}\ket{10}\). This leaves some entangled garbage, which we don't want!

This can be fixed in this case. Let \(C_2\) be the following circuit:
$$
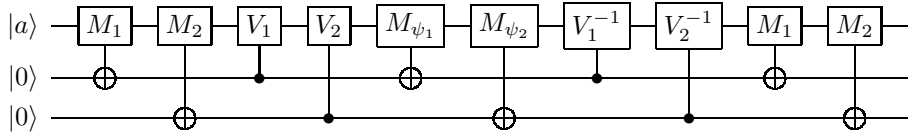
\Qcircuit @C=1em @R=.7em {
	\lstick{\ket{a}} & \gate{M_1} \qwx[1] & \gate{M_2} \qwx[2] & \gate{V_1} & \gate{V_2} & \gate{M_{\psi_1}} \qwx[1] & \gate{M_{\psi_2}} \qwx[2] & \gate{V_1^{-1}} & \gate{V_2^{-1}} & \gate{M_1} \qwx[1] & \gate{M_2} \qwx[2] & \qw & \\
	\lstick{\ket{0}} & \targ & \qw & \ctrl{-1} & \qw & \targ & \qw & \ctrl{-1} & \qw & \targ & \qw & \qw & \\
	\lstick{\ket{0}} & \qw & \targ & \qw & \ctrl{-2} & \qw & \targ & \qw & \ctrl{-2} & \qw & \targ & \qw & \\
}$$
Now this works if \(\ket{\psi_1}\) and \(\ket{\psi_2}\) are orthogonal to \(\ket{1}\) and \(\ket{2}\). But if they are not, then again this will leave us with entangled garbage!

Now, we could do something like detect if \(\ket{\psi_i}\) is in the subspace of \(\ket{1}\) and \(\ket{2}\) or orthogonal. But what if it is only partially in the subspace, such as if \(\ket{\psi_1} = \frac{1}{\sqrt{2}}\left(\ket{1} + \ket{3}\right)\) and \(\ket{\psi_2} =\frac{1}{\sqrt{2}}\left(\ket{1} - \ket{3}\right)\)? It seems less obvious how to handle these kinds of cases in a general way.

\subsubsection{A More Sophisticated, ``Almost'' Solution}

There is an efficient algorithm that comes very close to solving the unitarization problem. The insight is simple: we know how to efficiently swap orthogonal quantum states without disturbing anything else. This is true even if our algorithms to produce these states do not return ancilla bits to 0 on other inputs.

So just look for $k$ computational basis states that are close to orthogonal from all the $\ket{\psi_i}$. If there are enough dimensions we can't calculate a brute force solution, then this can be found by just measuring our states many times and choosing basis states that were never measured. Call these states we found that are almost orthogonal to $\ket{\psi_i}$ as $(\ket{a_i})_{i \in [k]}$. Then we can approximate a unitary that approximately swaps $\ket{a_i}$ with $\ket{\psi_i}$. And of course we can make a unitary that takes $\ket{i-1}$ to $\ket{a_i}$. Then name $U$ the circuit that first swaps $(\ket{i-1})_{i \in [k]}$ with $(\ket{a_i})_{i \in [k]}$, and then approximately swaps $(\ket{a_i})_{i \in [k]}$ with $(\ket{\psi_i})_{i \in [k]}$.

This algorithm is much more efficient than the algorithm presented in this paper, running in time polynomial in $\frac{k}{\epsilon}$. It also $\epsilon$-approximates a unitary and maps each $\ket{i-1}$ to  $\ket{\psi_i}$.

The issue is that this does not necessarily approximately leave orthogonal states unchanged. In particular, the state $\ket{a_1}$ may be orthogonal to each $(\ket{i-1})_{i \in [k]}$ in addition to $(\ket{\psi_i})_{i \in [k]}$. Then our algorithm first swaps it with $\ket{0}$, then where it ends up depends on how large the $\ket{0}$ component of $\ket{\psi_1}$ is. Unless $\ket{\psi_1}$ is close to $\ket{0}$, it won't end close to $\ket{a_1}$.

Then we are left with an issue with this algorithm that looks very similar to the first proposed algorithm.

\subsection{Notation}
\label{Notation}

We denote by $[k] = \{i \in \Z: i \geq 1, i \leq k\}$. If we write $\ket{i}$ for some natural $i$, we mean the $i$th computational basis state, represented in the obvious binary way. The number of bits that $ith$ state is written in is normally implicit. We may also sometimes write the $0$ computational basis state as $\ket{0^j}$ for some natural $j$ just to make explicit the number of qubits.

For our gate set, we have gates that
\begin{enumerate}
    \item Perform single qubit operatons.
    \item Swap $O(n)$ bit computational basis states.
    \item Apply $O(n)$ bit oracles.
    \item Apply $O(n)$ bit oracle inverses.
    \item Apply any of the above controlled on if some set of qubits have a given, computational basis state.
\end{enumerate}

Let us define approximation.
\begin{definition}[Unitary Approximation]
    For \(\epsilon > 0\), integers \(n\) and \(p\), unitary $U$ on $n$ qubits and unitary $C$ on $n + p$ qubits, we say \(C\) \(\epsilon\)-approximates $U$ if for all $n$ qubit states \(\ket{a}\) we have
    \[\|(U \ket{a})\ket{0^p} - C \ket{a}\ket{0^p}\| \leq \epsilon.\]
\end{definition}
The key here is that we let $C$ use ancilla qubits, as long as it \emph{always} returns them close to $0$. $0$ approximating means implementing $U$ exactly, given access to ancilla bits.

A straightforward consequence of this definition is that unitary approximations can be composed.
\begin{lemma}[Approximate Composition]
    Suppose for integer \(l\), we have for each \(i \in [l]\) a unitary \(V_i\) on $n$ qubits, and a unitary \(C_i\) on \(n + p_i\) qubits such that \(C_i\) \(\epsilon_i\)-approximates \(V_i\).
    
    Let \(V = \prod_{i \in [l]} V_i\), \(p = \max_{i \in [l]}(p_i)\), \(C = \prod_{i \in [l]} C_i \otimes I_{p-p_i}\), and \(\epsilon = \sum_{i \in [l]} \epsilon_i\).
    
    Then $C$ \(\epsilon\)-approximates $V$. 
\end{lemma}

Now to be concrete, let us define state preparation.

\begin{definition}[State Preparation]
    For a state $\ket{\psi}$ on $n$ qubits, and a unitary $V$ on $n + m$ qubits for some $p$, we say that $V$ prepares $\ket{\psi}$ if
    \[V\ket{0^n}\ket{0^m} = \ket{\psi}\ket{0^m}.\]
    
    We say $V$ $\epsilon$-approximately prepares $\ket{\psi}$ if $V$ $\epsilon$-approximates a unitary that prepares $\ket{\psi}$.
\end{definition}

Note we made no assertions on what $V$ does for any other inputs.

For notation, define the unitaries that swap computational basis states:
\begin{definition}[Basis Zero Swap]
\label{BasisToZero}
    For computational basis state \(\ket{a}\) let \(P_a\) be the unitary that swaps \(\ket{0}\) with \(\ket{a}\).
\end{definition}

We define the $\textmd{OR}$ and $M_0$ as:
\begin{definition}[Or and $M_0$]
\label{OrAndNor}
    For any gate $A$, the gate $M_0$ is defined so that the following circuit
    $$
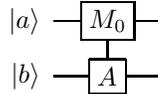
\Qcircuit @C=1em @R=.7em {
    	\lstick{\ket{a}} & \gate{M_0} \qwx[1] & \qw &  \\
    	\lstick{\ket{b}} & \gate{A} & \qw & \\
    }$$
    (in superposition) applies $A$ to register $b$ if and only if register $a$ is the state $\ket{0}$.
    
    Similarly, the gate $\textmd{OR}$ is defined so that
    $$
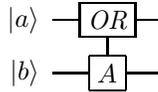
\Qcircuit @C=1em @R=.7em {
    	\lstick{\ket{a}} & \gate{\textmd{OR}} \qwx[1] & \qw &  \\
    	\lstick{\ket{b}} & \gate{A} & \qw & \\
    }$$
    (in superposition) applies $A$ to register $b$ if and only if register $a$ is \emph{not} the state $\ket{0}$.
\end{definition}

And now our control gates.

\begin{definition}[Controlled on Measurements]
\label{CMNot}
    For basis state \(\ket{a}\), define $M_a$ so that for any gate $A$:
    $$\Qcircuit @C=1em @R=.7em {
    	\lstick{\ket{\phi}} & \gate{M_a} \qwx[1] & \qw & = & & \lstick{\ket{\phi}} & \gate{P_a^{-1}} & \gate{M_0} \qwx[1] & \gate{P_a} & \qw & \\
    	\lstick{\ket{b}} & \gate{A} & \qw & \pureghost{=} & & \lstick{\ket{b}} & \qw  & \gate{A}  & \qw & \qw & \\
    }.$$
    
    Similarly for \(\ket{\psi}\) on $n$ qubits $\epsilon$-prepared by unitary $V$ on $n + p$ qubits, define \(M_{\psi}\) so that for any gate $A$ 
    $$\Qcircuit @C=1em @R=.7em {
    	\lstick{\ket{\phi}} & \gate{M_{\psi}} \qwx[2] & \qw & & & \lstick{\ket{\phi}} & \multigate{1}{V^{-1}} & \multigate{1}{M_0} & \multigate{1}{V} & \qw & \\
    	                                       &  &  & = & & \lstick{\ket{0^p}} & \ghost{V^{-1}} & \ghost{M_0} \qwx[1] & \ghost{V} & \qw & \\
    	\lstick{\ket{b}} & \gate{A} & \qw & \pureghost{=} & & \lstick{\ket{b}} & \qw & \gate{A} & \qw & \qw & \\
    }.$$
\end{definition}
Note that the ancilla register in our notation here is left implicit. However, one can observe that that if $v$ is a $0$ approximation, then this always returns that register to $0$. So this still gives a $2\epsilon$ approximation of the controlled application if $V$ $\epsilon$-approximately prepares $\ket{\psi}$.

Finally, to clarify what $\epsilon$ approximately computes an $\epsilon$ approximate basis in the abstract means, we define:
\begin{definition}[$\epsilon$ Approximate Basis]
    For any set of quantum states $X$, and set of orthogonal quantum states $Y$, we say that $Y$ is an $\epsilon$ approximate basis for $X$ if for all $x \in X$, $\Delta(Y, x) < \epsilon$.
    
    We say that a process, $A$, outputting quantum states, $Z$, $\delta$ approximately computes an $\epsilon$ approximate basis for $X$ if $Z$ is a set of states within $\delta$ of $Y$.
\end{definition}
Then with this language, \cref{ApproximateBasis} can be viewed as giving a circuit that with probabality $1- \gamma$ $\epsilon$ approximately computes an $\epsilon$ approximate basis.

\subsection{Proof Idea}
\label{proofIdea}

There are two main ideas in this paper: an orthogonalization procedure, and a well conditioned orthogonalization order.

At a high level, we orthogonalize a state $\ket{\psi}$ with respect to orthogonal states $\ket{\phi_1}, \ldots, \ket{\phi_k}$ by first estimating the amplitude of $\ket{\psi}$ along each $\ket{\phi_i}$.  Then, all in superposition, we prepare $\ket{\psi}$, and record in an ancilla register which of the $\ket{\phi_i}$ the state is. If it was any of these, we prepare $\ket{\psi}$ again. After many iterations, we know that the final state is ``almost'' the orthogonal component of $\ket{\psi}$, with something in the ancilla registers related to the amplitudes we estimated. We can fix the state we are ``almost in'' to just be the orthogonal component of $\ket{\psi}$, and since this state is close to our actual state, this will almost convert our state into the orthogonal component of $\ket{\psi}$.

The time of the above orthogonalization procedure depends on how close $\ket{\psi}$ is to the subspace spanned by the $\ket{\phi_i}$ and how long each of the $\ket{\phi_i}$ take to compute. More importantly, the error in the orthogonal component depends on the product of the error in the $\ket{\phi_i}$ and how long it takes to compute.

The straightforward way to orthogonalize $k$ states to within $\epsilon$ is to just just keep a running set of $\ket{\phi_i}$, and orthogonalize some new $\ket{\psi}$, skipping any $\ket{\psi}$ if it is already within $\epsilon$ of the computed basis. Then roughly speaking, if we internally try to estimate the true orthogonal components to within $\gamma$, at worst case, we might incur an extra poly in $\ln(1/\gamma)/\epsilon$ factor of error and time after every orthogonalization. This gives a total error polynomial in $\gamma (\ln(1/\gamma) /\epsilon)^k$ and time polynomial in $(\ln(1/\gamma) /\epsilon)^k$. For the error to be less than $\epsilon$, we need $\gamma < \epsilon^k$. This implies the time of the algorithm is polynomial in $(k /\epsilon)^k$.

For this to be polynomial, we need $k = O(\frac{\log(n)}{\log(\log(n))})$. But for $k = \Omega(\frac{\log(n)}{\log(\log(n))})$, we can only get $\epsilon = \poly(1/\log(n))$ in polynomial time.. This is a pretty severe limitation in $\epsilon$, we would like to get $\epsilon$ polynomially small in $n$. We get halfway there: to $\epsilon = 2^{-O(\sqrt{\log(n)})}$.

Our basic idea is to keep two sets of vectors we have already orthogonalized, $A$, and $B$. The idea is to first orthogonalize a new $\ket{\psi}$ with respect to $A$, then with respect to $B$. $B$ will hopefully be further from $\ket{\psi}$ and have more recent states in it, so we don't have to pay this extra $1/\epsilon$ factor on top of how long it takes to compute states in $B$.

Everytime we orthogonalize a state, we add it to $B$. If there is any $\ket{\psi}$ remaining that is far from $B$, orthogonalize that. Otherwise, add all of $B$ to $A$ and continue. The key here is that we cannot add $B$ to $A$ more than $O(\ln(1/\epsilon))$ times before every remaining state is within $\epsilon$ of the span of $A$.

\subsubsection{Orthogonalization Procedure}

Now to outline the orthogonalization procedure in more detail. First take a natural $k$, orthonomal states \((\ket{\phi_i})_{i \in [k]}\) prepared by unitaries \((U_i)_{i \in [k]}\) and state \(\ket{\psi}\) prepared by unitary \(V\).

\begin{enumerate}
    \item First, estimate the angles of \(\ket{\psi}\) to each of \((\ket{\phi_i})_{i \in [k]}\). This will also implicitly tell us \(\delta = \Delta(\ket{\phi_i})_{i \in [k]}, \ket{\psi})\). Luckily, if \(\delta\) is too small, we don't have to worry about this component in our approximate basis.
    
    Let \(\ket{\phi} = \Phi((\ket{\phi_i})_{i \in [k]}, \ket{\psi})\).
    
    \item Assume the input is \(\ket{0^n}\). We will be making a subroutine \(A\) that will within \(\epsilon\) take \(\ket{0^n}\) to \(\ket{\phi}\). $A$ will run on a primary register and several registers of $\lceil\log(k)\rceil$ ancilla qubits. $A$ first runs \(V\) on the primary register. Then for each $i \in [k]$, it will, controlled on if the primary register is \(\ket{\phi_i}\), set the ancilla register to state $\ket{i}$.
    
    The part of the state that has not written into the ancilla register is, informally, correct. For these other cases, we need to rotate the primary register back to $\ket{0^n}$. We will then repeat a similar process.
    
    Here is the circuit for this first step.
    
    $$
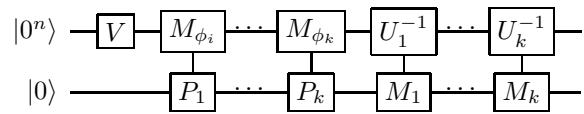
\Qcircuit @C=1em @R=.7em {
    	\lstick{\ket{0^n}} & \gate{V} & \gate{M_{\phi_i}} \qwx[1] & \qwelipse & \gate{M_{\phi_k}} \qwx[1] & \gate{U_1^{-1}} & \qwelipse & \gate{U_k^{-1}} & \qw &\\
    	\lstick{\ket{0}} & \qw & \gate{P_{1}}  & \qwelipse &  \gate{P_{k}} & \gate{M_1} \qwx[-1] & \qwelipse & \gate{M_k} \qwx[-1] &\qw &\\
    }$$
    At the end (step 4), we will need to come back and clean up these ancilla registers by rotating them back to $\ket{0}$, but we know what angle we need to rotate them by due to step 1.
    
    \item Then repeatedly, at step $j$, if the $j-1$ register is not 0, apply $V$ again, and record if the state is $\ket{\phi_i}$ again. The circuit for one iteration is:
    
    $$
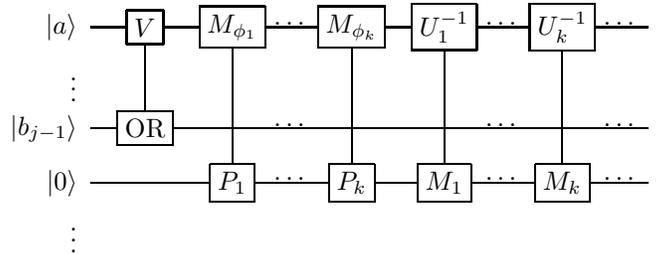
\Qcircuit @C=1em @R=.7em {
    	\lstick{\ket{a}} & \gate{V} & \gate{M_{\phi_1}} \qwx[3] & \qwelipse & \gate{M_{\phi_k}} \qwx[3] & \gate{U_1^{-1}} & \qwelipse & \gate{U_k^{-1}} & \qwelipse & \qw & \\
    	\lstick{\vdots} & \pureghost{\vdots} & \\
    	\lstick{\ket{b_{j-1}}} & \gate{\textmd{OR}} \qwx[-2] & \qw  & \qwelipse & \qw & \qw & \qwelipse & \qw & \qwelipse &  \qw &\\
    	\lstick{\ket{0}} & \qw & \gate{P_1} & \qwelipse & \gate{P_k}  & \gate{M_1} \qwx[-3] & \qwelipse & \gate{M_k} \qwx[-3] & \qwelipse & \qw &\\
    	\lstick{\vdots} & \pureghost{\vdots} & \\
    }.$$
    
    Let \(\eta = \sqrt{1 - \delta^2}\). Then the magnitude of the current state that has a non zero $j$th ancilla register is \(\eta^j\). Eventually, \(\eta^j < \epsilon\) for \(j = O \left(\frac{\ln(1/\epsilon)}{\delta^2}\right)\). And if the last ancilla register is $\ket{0}$, then the entangled primary register is exactly \(\ket{\phi}\).
    
    \item Now we need to clean up the ancilla bits. The idea is that we are going to directly rotate them back into place. So consider the ideal state that is exactly the same, except the primary register is exactly \(\ket{\phi}\). This has inner product \(1 - \eta^j\) with the actual state.
    
    Now, we fix the ideal state by rotating the bits back as if it was exactly this ideal state. Since we know the components from step 1, we can do this. This leaves us with the state \(\ket{\phi}\ket{0}^j\) on input \(\ket{0^n}\ket{0}^j\). Since the actual state we are acting on is only \(\epsilon\) off from the one we analyzed, the final state is also within \(\epsilon\).
    
    Let $R$ be a unitary such that
    \[R\ket{0} = \delta \ket{0} +  \sum_{i = 1}^k  \bra{\phi_i} \ket{\psi} \ket{i} .\]
    Then the final part of the circuit for $A$ is
    
    $$
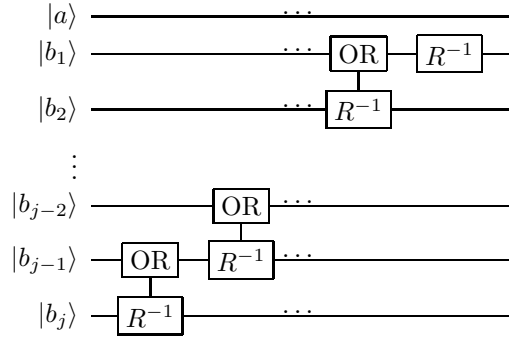
\Qcircuit @C=1em @R=.7em {
    	\lstick{\ket{a}} & \qw & \qw  & \qwelipse & \qw & \qw & \qw & \\
    	\lstick{\ket{b_1}} & \qw & \qw & \qwelipse & \gate{\textmd{OR}} \qwx[1] & \gate{R^{-1}} & \qw &\\
    	\lstick{\ket{b_2}} & \qw & \qw & \qwelipse & \gate{R^{-1}} & \qw & \qw &\\
    	\lstick{\vdots} & \pureghost{\vdots} & \\
    	\lstick{\ket{b_{j-2}}} & \qw & \gate{\textmd{OR}} \qwx[1] & \qwelipse & \qw & \qw & \qw &\\
    	\lstick{\ket{b_{j-1}}} & \gate{\textmd{OR}} \qwx[1] & \gate{R^{-1}} &\qwelipse & \qw & \qw & \qw &\\
    	\lstick{\ket{b_j}} & \gate{R^{-1}} & \qw & \qwelipse & \qw & \qw & \qw &\\
    }$$
    
    This completes subroutine A.
\end{enumerate}

One may also wonder why, when computing $A$, we do so in superposition, and not actually measure if we are in one of the $\ket{\phi_i}$. The reason why we cannot do this is the result is no longer unitary, and then we cannot control on it. We need $A$ to be unitary to perform $M_\psi$ from \cref{CMNot}.

\subsection{Previous Work}

This result is a refinement of ideas from a class project \cite{Coo19} and strictly improves those results. The specific technique in this paper, creating an approximate basis for several quantum states, seems new, though the basic techniques already existed in that earlier work.

Other work \cite{Kre20, ARU14, AKKT20, BR20} has also used inherently quantum oracles. In these works, for some state \(\ket{\psi}\), they use a quantum oracle, \(O_\psi\), that swaps \(\ket{\psi}\) with some known, computable state \(\ket{\perp}\), and changes no orthogonal state (or they use a similar, equivalent oracle). This is in contrast to the more common case where an oracle implements a classical algorithm, but may be invoked in super position.

Kretschmer \cite{Kre20} showed that \(O_\psi\) is equivalent to a a unitary that prepares \(\ket{\psi}\), but on a larger state space. Our work differs in that Kretschmer expands the state space of \(O_\psi\) to add \(\ket{\perp}\), while we insist on the final algorithm being unitary on the original space directly, and not some larger space.

Other works have been done on simulating unitaries based on a description of each entry in its matrix \cite{BC12}, or by simulating the hamiltonian of a unitary based on a description of its entries (\cite{BC12, BCK15}). In contrast, our work focuses on running a unitary where only the first few columns are given by black box oracles that produce them. These previous techniques scale polynomially with the sparsity of their matrices. Our technique doesn't rely on sparsity, but instead runs in time exponential in the number of dimensions of the subspace the unitary acts on.

Quantum speedup has also been used in various other linear algebra problems, for instance the famous HHL algorithm \cite{HHL09}. The HHL algorithm takes a sparse matrix $A$, and a quantum state $b$ and gives an efficient way to prepare the state $x$ so that $Ax = b$. The HHL algorithms still relies on the sparsity of matrices for efficiency.

The motivation of this work is to perform operations without computing garbage. Computation of unwanted, entangled garbage has been a problem in many other settings, such as QSampling. Broadly, these kinds of problems often fall into the paradigm of 'Quantum State Generation' \cite{AT03}. A common approach to generate quantum states is to express them as the ground state of a hamiltonian, then find it. Finding ground states of hamiltonians has been researched extensively (\cite{WY21}).

Our technique uses ideas from the swap test \cite{BCWD01} to estimate amplitudes of given quantum states along each other. The final algorithm is very similar to Gram-Schmidt orthogonalization.

There has been previous work on a problem called orthogonalization \cite{VAK13}, which refers to a completely different problem, with no known relations to our work. In that line of work, orthogonalization is the problem of given a quantum state $\ket{\psi}$, produces an orthogonal state $\ket{\perp}$. While there is no single unitary that can do this for all states $\ket{\psi}$, such work focuses on producing an orthogonal state using as little information about $\ket{\psi}$ as possible, and realizing it in physical systems.

\subsection{Paper Layout}

In section 2, we give some basic results about approximating components of a quantum state, and basic techniques for combining rotations to create arbitrary unitaries. In section 3, we formally give and prove our orthogonalization procedure, \cref{Orthogonalization}. In section 4, we go through the details required to construct our approximate basis, \cref{ApproximateBasis}, and then use it to prove \cref{CleanUnitary}.

\section{Sampling Angles and Rotations}

In this section, I show how you can approximate the components of a state from an oracle computing it. Then I show how to apply any unitary on a subspace we can compute all the basis vectors of with these angle approximations.

First, let us prove that we can estimate the components of one state along another with high probability.
\begin{lemma}[Angle Approximation]
\label{AngleApproximation}
    Let \(\ket{\psi}\) and \(\ket{\phi}\) be states on $n$ qubits prepared by $V$ and $U$, respectively.
    
    Then for every \(\epsilon, \gamma > 0\), there is a quantum circuit with size polynomial in \(O\left(\frac{\ln(1/\gamma)}{\epsilon^2} \right)\) making \(O\left(\frac{\ln(1/\gamma)}{\epsilon^2} \right) \) oracle calls to $V$ and $U$ computing complex number \(C\) with \(\|C - \bra{\phi}\ket{\psi}\| \leq \epsilon\) with probability at least \(1 - \gamma\).
\end{lemma}
\begin{proof}
    Let $m$ be the number of ancilla bits needed by $U$ and $V$.

    We will do something similar to a swap test. First we Hadamard an ancilla bit. Then conditioned on it being one, we apply \(V\) then \(U^{-1}\). Finally, we Hadamard the ancilla bit and measure both it and primary register.
    
    $$
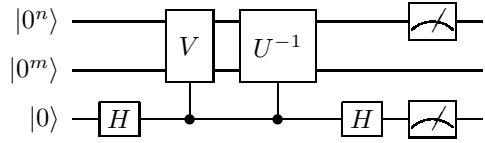
\Qcircuit @C=1em @R=.7em {
    	\lstick{\ket{0^n}} & \qw & \multigate{1}{V} & \multigate{1}{U^{-1}} & \qw & \meter &\qw & \\
    	\lstick{\ket{0^m}} & \qw & \ghost{V} & \ghost{U^{-1}} & \qw & \qw & \qw & \\
    	\lstick{\ket{0}} & \gate{H}  & \ctrl{-1} & \ctrl{-1} & \gate{H} & \meter &\qw & \\
    }$$
    
    Let \(C' = \bra{\phi}\ket{\psi}\). Then the probability of measuring a zero in the ancilla bit and \(\ket{0}\) in the first register is
    \[\frac{\|1 + C'\|^2}{4} = \frac{1}{4}\left(\textmd{complex}(C')^2 + (1 + \textmd{real}(C'))^2\right).\]
    Similarly, the probability of measuring a one in the ancilla bit and \(\ket{0}\) in the register is
    \[\frac{\|1 - C'\|^2}{4} = \frac{1}{4}\left(\textmd{complex}(C')^2 + (1 - \textmd{real}(C'))^2\right).\]
    Taking the difference of these probabilities directly gives us the real part of \(C'\). By adding a phase shift by \(i\) on \(V\), we can similarly estimate the complex component of \(C'\).
    
    Using Chernoff bounds, with probability at least \(1 - \gamma\), we can approximate \(C'\) within \(\epsilon\) using only \(O\left(\frac{\ln(1/\gamma)}{\epsilon^2} \right) \) calls to $V$.
\end{proof}

We also need a slightly stronger version of this theorem, saying that if we can approximately compute a state \(\ket{\phi}\), then we can approximately compute the inner product of \(\ket{\psi}\) and \(\ket{\phi}\).
\begin{lemma}[Robust Angle Approximation]
\label{RobustAngleApproximation}
    Let \(\ket{\psi}\) be a state on $n$ qubits prepared by $V$. Let \(\ket{\phi}\) be a state $\eta$-prepared by $U$.
    
    Then for every \(\epsilon, \gamma > 0\), there is a quantum circuit with size polynomial in \(O\left(\frac{\ln(1/\gamma)}{\epsilon^2} \right)\) making \(O\left(\frac{\ln(1/\gamma)}{\epsilon^2} \right) \) oracle calls to $V$ and $U$ computing complex number \(C\) with \(\|C - \bra{\phi}\ket{\psi}\| \leq \epsilon + \eta\) with probability at least \(1 - \gamma\).
\end{lemma}
\begin{proof}
    Running \cref{AngleApproximation}, with probability at least \(1 - \gamma\), we will get a \(C\) so that
    \[|C - (\bra{\psi}\bra{0^m})U\ket{0^n}\ket{0^m}| \leq \epsilon.\]
    Then
    \begin{align*}
        |C - \bra{\psi}\ket{\phi}| = & |C - (\bra{\psi}\bra{0^m})U\ket{0^n}\ket{0^m} + (\bra{\psi}\bra{0^m})U\ket{0^n}\ket{0^m} - \bra{\psi}\ket{\phi}| \\
        \leq & |C - (\bra{\psi}\bra{0^m})U\ket{0^n}\ket{0^m}| + |(\bra{\psi}\bra{0^m})U\ket{0^n}\ket{0^m} - (\bra{\psi}\bra{0^m})\ket{\phi}\ket{0^m}| \\
        \leq & \epsilon + |U\ket{0^n}\ket{0^m} - \ket{\phi}\ket{0^m}| \\
        \leq & \epsilon + \eta .
    \end{align*}
\end{proof}

After we have an approximate basis, the basic idea is to move the components of the input state in that basis into ancilla bits, do the local transformation on those components exactly and then move it back.

To do this, we need to convert our approximations of the components in a basis to a unitary on that basis.
\begin{restatable}[Gram-Schmidt Is Stable]{lemma}{GramSchmidtIsStab}
\label{GramSchmidtIsStab}
    Let $(v_i)_{i \in [k]}$ be orthonormal vectors. Let $\epsilon \in (0, \frac{1}{64 k^2})$. Suppose $(w_i)_{i \in [k]}$ are vectors such that for each $i \in [k]$, $\|v_i - w_i\| \leq \epsilon$. Then running Gram Schmidt on $(w_i)_{i \in [k]}$ produces orthonormal \((u_i)_{i \in [k]}\), where for each $i \in [k]$, 
    \[\|v_i - u_i\| \leq (64k + 1) \epsilon.\]
\end{restatable}

I provide the proof of \cref{GramSchmidtIsStab} in the appendix. A direct corollary of this is that we can extrapolate approximations of unitaries.
\begin{corollary}[Extrapolate Unitary]
\label{ExtrapolateUnitary}
     Let $(v_i)_{i \in [k]}$ be orthonormal vectors in $m$ dimensions. Let $\epsilon \in (0, \frac{1}{64 k^2})$. Suppose $(w_i)_{i \in [k]}$ are vectors in $m$ dimensions such that for each $i \in [k]$, $\|v_i - w_i\| \leq \epsilon$. Then there is a \(\poly(m)\) time algorithm that constructs an \(m \times m\) unitary matrix from $(w_i)_{i \in [k]}$ who's first $k$ columns are \((u_i)_{i \in [k]}\), where for each $i \in [k]$, $\|v_i - u_i\| \leq (64k + 1) \epsilon$.
\end{corollary}

It is a standard result that quantum circuits can implement any given unitary. I omit this proof.
\begin{lemma}[Apply Arbitrary Unitary]
\label{ApplyArbitraryUnitary}
    For any unitary $U$ on $\log_2(m)$ qubits, there is a quantum circuit of size \(\poly(m)\) implementing $U$ computable in polynomial time.
\end{lemma}

Given oracles to prepare an orthogonal basis for \(\ket{\psi_i}\) and the components of our \(\ket{\psi_i}\) along them, we can compute the unitary from our main theorem using only a constant number of oracle calls.

\begin{lemma}[Deferred Rotation]
\label{DefferedRotation}
    Let $k$ and $l$ be natural numbers. Let \((\ket{\phi_i})_{i \in [l]}\) be $l$ orthonormal vectors all orthogonal to \((\ket{i-1})_{i \in [k]}\). And for each \(i \in [l]\), let \(V_i\) be a unitary that prepares $\ket{\phi_i}$.
    
    Consider some orthonormal \((\ket{\psi_i})_{i \in [k]}\). Let \(m = k + l\). Suppose that for each \(i \in [k]\), we know the $m$ dimensional vector \(v^i\) so that \(\ket{\psi_i} = \sum_{j \in [k]}v^i_j \ket{j-1} + \sum_{j \in [l]}v^i_{k-1 + j} \ket{\phi_j}\).
    
    Then in polynomial time, we can construct a circuit with size \(\poly(m)\) making a constant number of oracle calls to each \(V_i\) $0$-approximating a unitary, $U$, on \(n\) bits such that for all \(i \in [k]\), \(U\ket{i-1} = \ket{\psi_i}\) and for all \(\ket{\theta}\) orthogonal to all \((\ket{i})_{i \in [n]}, (\ket{\phi_i})_{i \in [l]}\), we have \(U\ket{\theta} = \ket{\theta}\).
\end{lemma}
\begin{proof}
    We will use \(p = \lceil \log_2(m + 1) \rceil\) auxiliary bits to hold the \((\ket{i-1})_{i \in [k]}\), \((\ket{\phi_i})_{i \in [l]}\) subspace. Then we will perform the rotations on this subspace, then we will move it back.
    
    First we construct the circuit, \(A_1\), which puts the components of the state in the first $k$ computational basis states into the ancilla state, and similarly for the states \((\ket{\phi_i})_{i \in [l]}\). So \(A_1\) is the following circuit:
    $$
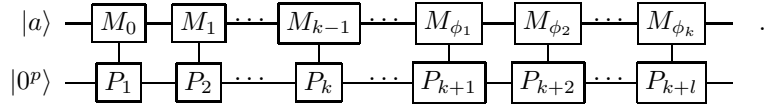
\Qcircuit @C=1em @R=.7em {
    	\lstick{\ket{a}}  & \gate{M_0} \qwx[1]  & \gate{M_1} \qwx[1] & \qwelipse & \gate{M_{k-1}} \qwx[1] & \qwelipse & \gate{M_{\phi_1}} \qwx[1]  & \gate{M_{\phi_2}} \qwx[1] & \qwelipse & \gate{M_{\phi_k}} \qwx[1]& \qw & \\
    	\lstick{\ket{0^p}} & \gate{P_1}  & \gate{P_2} & \qwelipse & \gate{P_k} & \qwelipse & \gate{P_{k+1}} \qwx[-1] & \gate{P_{k+2}} \qwx[-1] & \qwelipse & \gate{P_{k+l}} \qwx[-1]& \qw & \\
    } .$$
    
    Then controlled on the ancilla state, we move our subspace in the primary register all back to the 0 state. So $A_2$ is the following circuit:
    $$
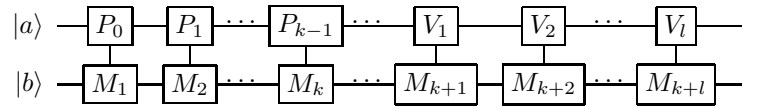
\Qcircuit @C=1em @R=.7em {
    	\lstick{\ket{a}} & \gate{P_0}  & \gate{P_1} & \qwelipse & \gate{P_{k-1}} & \qwelipse & \gate{V_1} & \gate{V_2} & \qwelipse & \gate{V_l} & \qw & \\
    	\lstick{\ket{b}} & \gate{M_1} \qwx[-1]  & \gate{M_2} \qwx[-1] & \qwelipse & \gate{M_k} \qwx[-1] & \qwelipse & \gate{M_{k+1}} \qwx[-1]  & \gate{M_{k+2}} \qwx[-1] & \qwelipse & \gate{M_{k+l}} \qwx[-1]& \qw & \\
    } .$$
    
    Then using \cref{ExtrapolateUnitary}, we can get a $m \times m$ unitary, \(U'\), who's first $k$ column vectors are \((v_i)_{i \in [k]}\). Of course by padding $U'$, we can get unitary $U''$ on $p$ qubits. Particularly, pad it so that $U''$ is identity on everything but states $\ket{1}, \ldots, \ket{m}$. Notably, it is identity on the $\ket{0}$ state.
    
    Then using \cref{ApplyArbitraryUnitary}, there is a circuit with size \(O(2^p) = O(m)\) implementing $U''$, call it $A_3$. Then $A_3$ does the rotation on the ancilla register.
    
    Finally, we just run \(A_2^{-1}\) and \(A_1^{-1}\) to translate our results back into our main state. This leaves us with the final circuit for $U$:
    
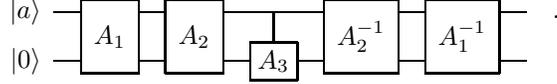
    $$\Qcircuit @C=1em @R=.7em {
    	\lstick{\ket{a}} & \multigate{1}{A_1} & \multigate{1}{A_2} & \qw \qwx[1] & \multigate{1}{A_2^{-1}} & \multigate{1}{A_1^{-1}} & \qw & \\
    	\lstick{\ket{0}} & \ghost{A_1} & \ghost{A_2} & \gate{A_3} & \ghost{A_2^{-1}} & \ghost{A_1^{-1}} & \qw & \\
    } .$$
    
    Now to see if this approximates an appropriate $U$, we just need to see if it gives the appropriate results on some basis.
    \begin{itemize}
        \item \(i+1 \in [k]\), \(\ket{i}\). This gives
        \begin{align*}
            A_1 \ket{i} \ket{0} = & \ket{i} \ket{i+1} \\
            A_2 A_1 \ket{i} \ket{0} = &  \ket{0} \ket{i+1} \\
            (I^{\otimes n} \otimes U'') A_2 A_1 \ket{i} \ket{0} = &  \ket{0} \sum_{j \in [m]} v^i_j \ket{j} \\
            A_2^{-1} (I^{\otimes n} \otimes U'') A_2 A_1 \ket{i} \ket{0} = & \sum_{j \in [k]}v^i_j \ket{j-1} \ket{j} + \sum_{j \in [l]}v^i_{k + j} \ket{\phi_j} \ket{k+j} \\
            A_1^{-1} A_2^{-1} (I^{\otimes n} \otimes U'') A_2 A_1 \ket{i} \ket{0} = & \sum_{j \in [k]}v^i_j \ket{j-1} \ket{0} + \sum_{j \in [l]}v^i_{k + j} \ket{\phi_j} \ket{0} \\
            = & \ket{\psi_i} \ket{0}. \\
        \end{align*}
        
        As needed for $U$, takes \(\ket{i}\) to \(\ket{\psi_i}\) and leaves auxiliary bit as $0$.
        
        \item \(i \in [k]\), \(\ket{\phi_{i}}\). For notation, let $U''^i_j$ be the $j$th index of the $i$th column of $U''$. This gives
        \begin{align*}
            A_1 \ket{\phi_i} \ket{0} = & \ket{\phi_i} \ket{k+i} \\
            A_2 A_1 \ket{\phi_i} \ket{0} = &  \ket{0} \ket{k+i} \\
            (I^{\otimes n} \otimes U'') A_2 A_1 \ket{\phi_i} \ket{0} = &  \ket{0} \sum_{j \in [m]} U''^i_j \ket{j} \\
            A_2^{-1} (I^{\otimes n} \otimes U'') A_2 A_1 \ket{\phi_i} \ket{0} = & \sum_{j \in [k]} U''^{k+i}_j \ket{j-1} \ket{j} + \sum_{j \in [l]} U''^{k+i}_{k + j} \ket{\phi_j} \ket{k+j} \\
            A_1^{-1} A_2^{-1} (I^{\otimes n} \otimes U'') A_2 A_1 \ket{\phi_i} \ket{0} = & \sum_{j \in [k]} U''^{k+i}_j \ket{j-1} \ket{0} + \sum_{j \in [l]} U''^{k+i}_{k + j} \ket{\phi_j} \ket{0} \\
            = & \left(\sum_{j \in [k]} U''^{k+i}_j \ket{j} + \sum_{j \in [l]} U''^{k+i}_{k + j} \ket{\phi_j}\right) \ket{0}. \\
        \end{align*}
        
        This leaves the auxiliary bits at $0$, as required.
        
        \item \(\ket{\theta}\) for any state orthogonal to all \((\ket{i})_{i \in [k]}\) and all \((\ket{\phi_i})_{i \in [l]}\). This gives
        \begin{align*}
            A_1 \ket{\theta} \ket{0} = & \ket{\theta} \ket{0} \\
            A_2 A_1 \ket{\theta} \ket{0} = &  \ket{\theta} \ket{0} \\
            (I^{\otimes n} \otimes U'') A_2 A_1 \ket{\theta} \ket{0} = &  \ket{\theta} \ket{0} \\
            A_2^{-1} (I^{\otimes n} \otimes U'') A_2 A_1 \ket{\theta} \ket{0} = & \ket{\theta}\ket{0} \\
            A_1^{-1} A_2^{-1} (I^{\otimes n} \otimes U'') A_2 A_1 \ket{\theta} \ket{0} = & \ket{\theta}\ket{0}. \\
        \end{align*}
        
        Which leaves the state unchanged.
    \end{itemize}
    
    So the circuit does the correct thing on every basis vector, so it indeed implements the appropriate operation.
\end{proof}

\section{How To Orthogonalize}

Next we prove the existence of our subroutine \cref{Orthogonalization}, described as $A$ in \cref{proofIdea}. This is the subroutine that takes an oracle preparing an input to a circuit that approximately prepares the orthogonal component of that state. After that, its a small step to make it prepare the orthogonal component.

\Orthogonalization*
Importantly here, the circuit ONLY uses the relationship between \(\ket{\psi}\) and the \(\ket{\phi_i}\) in $R$. So if we give this a subroutine an \(\eta\) approximation of \(R\), that only contributes \(\eta  O \left(\frac{ \ln(1/\epsilon)}{\delta^2}\right)\) error.
\begin{proof}
    First, assume $\epsilon < 1/2$. If not, solve for $\epsilon = 1/4$, and the same result works for $\epsilon$. If $k = 0$, $V$ already prepares $\ket{\psi}$, which is orthogonal to every element of $\emptyset$. So take $k \geq 1$.

    The algorithm is the same as in \cref{proofIdea} . Start with \(\ket{0}\), apply \(V\). Controlled on it being in one of the \(k\) states, set an ancilla register. Then conditioned on the ancilla register being set, rotate the state back to \(\ket{0}\). Repeat until all but a small amplitude has the final ancilla bits set to 0. Then rotate the ancilla bits back to 0.
    
    Let $m$ be the number of ancilla bits used by $V$.
    
    Then the circuit first does:
    
    $$
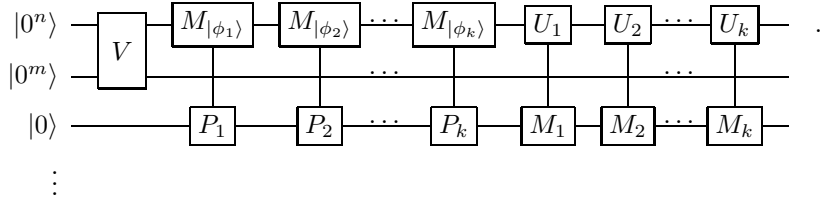
\Qcircuit @C=1em @R=.7em {
    	\lstick{\ket{0^n}} & \multigate{1}{V} & \gate{M_{\ket{\phi_1}}} \qwx[2] & \gate{M_{\ket{\phi_2}}} \qwx[2] & \qwelipse & \gate{M_{\ket{\phi_k}}} \qwx[2] & \gate{U_1} & \gate{U_2} & \qwelipse & \gate{U_k} & \qw & \\
    	\lstick{\ket{0^m}} & \ghost{V} & \qw & \qw & \qwelipse & \qw & \qw & \qw & \qwelipse & \qw & \qw & \\
    	\lstick{\ket{0}} & \qw & \gate{P_1}  & \gate{P_2} & \qwelipse &  \gate{P_k} &  \gate{M_1} \qwx[-2] & \gate{M_2} \qwx[-2] & \qwelipse & \gate{M_k} \qwx[-2]& \qw  &\\
    	\lstick{\vdots} & \pureghost{\vdots} & \\
    }.$$
    Call this \(A_1\).
    
    Let \(l = 4 \frac{\ln(1/\epsilon)}{\delta^2}\). Then for \(j \in \left[l - 1 \right]\) (going from $1$ to $l - 1$), each with its own k state register, we do
    $$
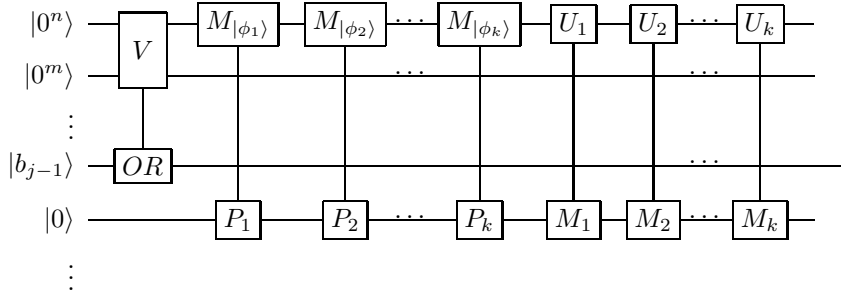
\Qcircuit @C=1em @R=.7em {
    	\lstick{\ket{0^n}} & \multigate{1}{V} & \gate{M_{\ket{\phi_1}}} \qwx[4] & \gate{M_{\ket{\phi_2}}} \qwx[4] & \qwelipse & \gate{M_{\ket{\phi_k}}} \qwx[4] & \gate{U_1} & \gate{U_2} & \qwelipse & \gate{U_k} & \qw & \\
    	\lstick{\ket{0^m}} & \ghost{V} & \qw & \qw & \qwelipse & \qw & \qw & \qw & \qwelipse & \qw & \qw & \\
    	\lstick{\vdots} & \pureghost{\vdots} & \\
    	\lstick{\ket{b_{j-1}}} & \gate{OR} \qwx[-2] & \qw  & \qw & \qw & \qw & \qw & \qw & \qwelipse & \qw & \qw & \qw &\\
    	\lstick{\ket{0}} & \qw & \gate{P_1}  & \gate{P_2} & \qwelipse &  \gate{P_k} &  \gate{M_1} \qwx[-4] & \gate{M_2} \qwx[-4] & \qwelipse & \gate{M_k} \qwx[-4]& \qw  &\\
    	\lstick{\vdots} & \pureghost{\vdots} & \\
    }.$$
    Call this entire sequence \(A_2\).

    For \(j \in \left[l -1 \right]\) (from $l - 1$ to $1$), we reset the auxiliary qubits.
    $$
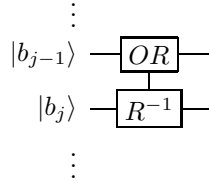
\Qcircuit @C=1em @R=.7em {
    	\lstick{\vdots} & \pureghost{\vdots} & \\
    	\lstick{\ket{b_{j-1}}} & \gate{OR} \qwx[1] &  \qw &\\
    	\lstick{\ket{b_j}} &  \gate{R^{-1}}  & \qw  &\\
    	\lstick{\vdots} & \pureghost{\vdots} & \\
    }.$$
    Call this entire sequence \(A_3\).
    
    Finally, we just have to clean up the first set of auxiliary bits:
    $$
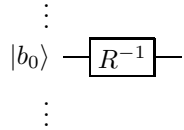
\Qcircuit @C=1em @R=.7em {
    	\lstick{\vdots} & \pureghost{\vdots} & \\
    	\lstick{\ket{b_{0}}} & \gate{R^{-1}} &\qw &\\
    	\lstick{\vdots} & \pureghost{\vdots} & \\
    }.$$
    Call this \(A_4\).
    
    Then \(A = A_4 A_3 A_2 A_1\), and \(p = m + O(\log(k + 1) ) l\).  First, I will show that \(A_2 A_1 \ket{0^n}\ket{0^p}\) approximately prepares some state, \(\ket{a}\). Then I will show that \(A_4 A_3 \ket{a} = \ket{\phi} \ket{0^p}\). This implies $A$ approximately prepares \(\ket{\phi}\ket{0^p}\).
    
    Let \(\alpha = \sqrt{1 - \beta}\) and
    \[\ket{\theta} = \sum_{i \in [k]}  \frac{\bra{\phi_i}\ket{\psi}}{\alpha} \ket{i}\]
    so that $R\ket{0} = \beta \ket{0} + \alpha \ket{\theta}$.
    
    Let
    \[\ket{b} = \beta \sum_{i = 0}^{l-1} \alpha^i \ket{\phi}\ket{0^m} \ket{\theta}^{i} \ket{0}^{l-i} + \alpha^{l} \ket{0^n}\ket{0^m}\ket{\theta}^{l} \]
    and
    \[\ket{a} = \ket{\phi}\ket{0^m} \left(\beta \sum_{i = 0}^{l-1} \alpha^i \ket{\theta}^{i} \ket{0}^{l-i}  + \alpha^{l} \ket{\theta}^{l} \right) . \]
    Then see that:
    \begin{align*}
        \|\ket{a} - \ket{b}\| \leq & 2 \alpha^l \\
        = & 2 (\alpha^2)^{2\frac{\ln(1/\epsilon)}{\delta^2}} \\
        = & 2 (1 - \beta^2)^{2\frac{\ln(1/\epsilon)}{\delta^2}} \\
        < & 2 (\epsilon)^{2\frac{\beta^2}{\delta^2}} \\
        \leq & 2 (\epsilon)^{2} \\
        \leq & \epsilon.
    \end{align*}
    This uses that $(1 - x) < e^{-x}$ for positive $x$.
    
    I claim \(A_2 A_1 \ket{0^n}\ket{0^p} = \ket{b}\) and \(A_4 A_3 \ket{a} = \ket{\phi} \ket{0^p}\). Since distances are preserved under unitary operations and \(\ket{b}\) is within \(\epsilon\) of \(\ket{a}\), this would prove the claim.
    
    See that
    \[A_1 \ket{0^n}\ket{0^p} = \beta\ket{\phi}\ket{0^m} \ket{0}^l + \alpha \ket{0^n}\ket{0^m}\ket{\theta}\ket{0}^{l-1}.\]
    Then one can show by induction after the $j$th step of \(A_2\), the state is
    \[\beta \sum_{i = 0}^{j} \alpha^i \ket{\phi}\ket{0^m} \ket{\theta}^{i} \ket{0}^{l-i} + \alpha^{j+1} \ket{0^n}\ket{0^m}\ket{\theta}^{j+1}\ket{0}^{l-j-1},\]
    which then at the \(j=l-1\) step gives \(A_2 A_1 \ket{0^n}\ket{0^p} = \ket{b}\).
    
    To prove \(A_4 A_3 \ket{a} = \ket{\phi} \ket{0^p}\) follows by basically the same induction. See that at the jth step of \(A_3\), the state is 
    \begin{align*}
        &\ket{\phi}\ket{0^m} \left(\beta \sum_{i = 0}^{j} \alpha^i \ket{\theta}^{i} \ket{0}^{l-i} + \alpha^{j+1} \ket{\theta}^{j+1}\ket{0}^{l-j-1}\right) \\
        = & \ket{\phi}\ket{0^m} \left(\beta \left(\sum_{i = 0}^{j-1} \alpha^i \ket{\theta}^{i} \ket{0}^{l-i} + \alpha^j \beta\ket{\theta}^{j}\ket{0}^{l-j} \right) + \alpha^{j+1} \ket{\theta}^{j+1} \ket{0}^{l-j-1} \right). \\
        = & \ket{\phi}\ket{0^m} \left(\beta \sum_{i = 0}^{j-1} \alpha^i \ket{\theta}^{i} \ket{0}^{l-i} + \alpha^j \ket{\theta}^{j} \left( \beta\ket{0} + \alpha\ket{\theta} \right)\ket{0}^{l-j-1} \right). \\
    \end{align*}
    So when we apply the conditional rotation essentially undoes one rotation to give
    \[\ket{\phi}\ket{0^m} \left(\beta \sum_{i = 0}^{j-1} \alpha^i \ket{\theta}^{i} \ket{0}^{l-i} + \alpha^j\ket{\theta}^{j}\ket{0}^{l-j})\right).\]
    And similarly \(A_4\) does the final, non conditional step to convert
    \[\ket{\phi}\ket{0^m} \left(\beta \ket{0}^l + \alpha \ket{\theta}\ket{0}^{l-1}\right)\]
    into \(\ket{\phi}\ket{0^p}\). Thus \(A_4 A_3 \ket{a} = \ket{\phi} \ket{0^p}\).
    
    So then
    \begin{align*}
        \|A_4 A_3 A_2 A_1 \ket{0^n}\ket{0^p} - \ket{\phi}\ket{0^p}\| = & \| A_4 A_3 \ket{b} -  \ket{\phi}\ket{0^p}\| \\
        = & \| A_4 A_3 \ket{b} - A_4 A_3 \ket{a} + A_4 A_3 \ket{a} -   \ket{\phi}\ket{0^p}\| \\
        = & \| A_4 A_3 (\ket{b} - \ket{a})\| \\
        = & \|\ket{b} - \ket{a}\| \\
        \leq & \epsilon .\\
    \end{align*}
    Finally, observe we only call $V$, each \(U_i\) and $R$ $O(l)$ times. Thus proves the statement.
\end{proof}

\section{Computing an Approximate Orthonormal Basis}

Our main lemma, \cref{Orthogonalization}, together with \cref{DefferedRotation} is enough to prove a version of our main theorem, \cref{CleanUnitary}. Still, the proofs are quite involved because we iteratively use approximations to find approximations and we have to be careful to make sure the errors don't accumulate in an unmanageable way.

The idea is straightforward. Iteratively find \(\ket{\phi_i}\) as the component of \(\ket{\psi_i}\) orthogonal to all \((\ket{\psi_j})_{j < i}\) using \cref{Orthogonalization}. If the length of \(\ket{\psi_i}\) along \(\ket{\phi_i}\) is small, then don't calculate \(\ket{\phi_i}\). 

First, we will quantify the error in \cref{RobustAngleApproximation} and \cref{Orthogonalization} when we use approximate preparation of states instead of exact ones. We put these in the same theorem because the theorem statements are always used together and there is a lot of redundancy in their statements.
\begin{lemma}[Orthogonal Components With Approximate Inputs]
\label{ApproximateOrthogonalComponents}
    For natural $k$, let \((\ket{\phi_i})_{i \in [k]}\) be orthonormal states on $n$ qubit and for every \(i \in [k]\), \(U_i\) $\eta_i$-approximately prepares $\ket{\phi_i}$. Let \(\eta = \sum_{i \in [k]} \eta_i\).
    
    Suppose \(\ket{\psi}\) is an $n$ qubit state prepared by $V$. Let \(\ket{\phi} = \Phi((\ket{\phi_i})_{i \in [k]}, \ket{\psi})\). Suppose for \(\delta \geq 0\), \(\Delta((\ket{\phi_i})_{i \in [k]}, \ket{\psi}) \geq \delta\).
    
    Then for \(\epsilon_1, \gamma > 0\), there is a circuit with size \(O \left(k \frac{\ln(k/\gamma) }{\epsilon_1^2} \right)\) making oracle calls to the \(U_i\) and \(V\) that outputs a constant \(C\), and for every \(i \in [k]\), a constant \(A_i\) such that
    \begin{enumerate}
        \item For \(\ket{\theta} = \sum_{i \in k} A_i \ket{\phi_i} + C \ket{\phi}\),  \(\ket{\theta}\) is a unit vector.
        \item Define \(\epsilon_2 \geq 0\) so that if \(2(\eta + \epsilon_1) < \delta^2\), 
        \[\epsilon_2 = \min\{4\frac{\eta + \epsilon_1}{\delta}, 2\sqrt{(\eta + \epsilon_1)}\},\]
        and otherwise:
        \[\epsilon_2 = 2\sqrt{(\eta + \epsilon_1)}.\]
        \item We call an output nice if for all $i \in [k]$,
        \[\|A_i - \bra{\phi_i}\ket{\psi}\| \leq \eta_i + \frac{\epsilon_1}{k}.\]
        An output is nice with probability at least \(1 - \gamma\).
        \item We call an output good if
        \[\|\ket{\theta} - \ket{\psi}\| \leq \epsilon_2 .\]
        An output is good if it is nice.
    \end{enumerate}
    
    Additionally, given good output from the first circuit and for $\delta > 0$, for all \(\epsilon_3 > 0\), there is a circuit $A$ such that
    \begin{enumerate}
        \item $A$ has size \(O(k \frac{\ln(1/\epsilon_3)}{\delta^2})\).
        \item $A$ is computable in polynomial time in its size.
        \item $A$ makes at most \(O( \frac{\ln(1/\epsilon_3)}{\delta^2})\) oracle calls to \(V\).
        \item For each i, $A$ makes at most \(O( \frac{\ln(k/\epsilon_3)}{\delta^2})\) oracle calls to \(U_i\).
        \item $A$
        \[O \left(\epsilon_2 \frac{\ln(1/\epsilon_3)}{\delta^2}\right) + \epsilon_3\]
        approximately prepares $\ket{\phi}$.
    \end{enumerate} 
\end{lemma}
\begin{proof}
The idea is to take the angle estimation algorithm in \cref{RobustAngleApproximation} and run it for each \(i \in [k]\) to get an \(\eta_i + \frac{\epsilon_1}{k}\) approximation of \(\bra{\phi_i}\ket{\psi} = A_i\) with probability \(1 - \frac{\gamma}{k}\). By a union bound, the output is nice with probability at least \(1 - \gamma\). Then we can calculate \(|C|\) by taking it so the squares of \(A\) and \(C\) sum to one. That is, we use \(A\) to estimate $C^2$ and take the square root. If the sums of squares of $A$ are greater than one, then let \(C = 0\) and normalize.

Suppose the output is nice. Then we have a couple cases on the calculated $C$:
\begin{itemize}
    \item[$C\neq0$] First lets look at the calculated difference.
    \begin{align*}
        \|\ket{\theta} - \ket{\psi}\| =& \|\sum_{i \in k} A_i \ket{\phi_i} + C \ket{\phi} - \ket{\psi}\| \\
        = & \sqrt{\sum_{i \in k} \|A_i - \bra{\phi_i}\ket{\psi}\|^2 + \|C - \bra{\phi}\ket{\psi}\|^2} \\
        \leq & \sqrt{\sum_{i \in k}(\eta_i + \frac{\epsilon_1}{k})^2 + |C - \bra{\phi}\ket{\psi}|^2} \\
        \leq & \sqrt{(\eta + \epsilon_1)^2 + |C - \bra{\phi}\ket{\psi}|^2}.
    \end{align*}
    
    Now we need to bound \(|C - \bra{\phi}\ket{\psi}|^2\). First we will switch to some vector notation for simplicity. Let \(u = (A_1, ...,A_k)\) and \(v = (\bra{\phi_1}\ket{\psi}, ..., \bra{\phi_k}\ket{\psi})\). Then we can write
    \begin{align}
        C =& \sqrt{1 - \|u\|^2} \\
        \Delta((\ket{\phi_i})_{i \in [k]}, \ket{\psi}) = \bra{\phi}\ket{\psi} =& \sqrt{1 - \|v\|^2} \geq \delta. \label{vdeltaBound}
    \end{align}
    
    For each \(i\), we have that \(|v_i - u_i| \leq \eta_i + \frac{\epsilon_1}{k}\). By a tailor approximation for \(x^2\), for any \(a, b \in (-1, 1)\), \(|a^2 - b^2| \leq 2|a - b|\). This gives \(|v_i^2 - u_i^2| \leq 2(\eta_i + \frac{\epsilon_1}{k})\). By a triangle inequality, we have that 
    \begin{equation} \label{SquareBound}
        |\|u\|^2 - \|v\|^2| \leq 2\sum_{i \in [k]}\left(\eta_i + \frac{\epsilon_1}{k}\right) = 2\left(\eta + \epsilon_1\right).
    \end{equation}
    
    So we really are comparing difference of square roots in terms of the differences of their arguments:
    \begin{align*}
        |C - \bra{\phi}\ket{\psi}| = & |\sqrt{1 - \|u\|^2} - \sqrt{1 - \|v\|^2}| .
    \end{align*}
    
    \begin{itemize}
    \item 
        First lets handle the general case with no lower bound on \(\bra{\phi}\ket{\psi}\). Since the square root function is positive, monotone increasing, and concave, we have for all \(a, b \geq 0\), \(|\sqrt{b} - \sqrt{a}| \leq \sqrt{|b-a|}\). Thus
        \begin{align*}
            |C - \bra{\phi}\ket{\psi}| = & |\sqrt{1 - \|u\|^2} - \sqrt{1 - \|v\|^2}| \\
            \leq & \sqrt{|\|u\|^2 - \|v\|^2|} \\
            \leq & \sqrt{2(\eta + \epsilon_1)}.
        \end{align*}
        If \(\eta + \epsilon_1 > 1\), then trivially we have \(\|\ket{\theta} - \ket{\psi}\| \leq 2 \leq 2\sqrt{\eta + \epsilon_1}\). Otherwise:
        \begin{align*}
            \|\ket{\theta} - \ket{\psi}\| = & \sqrt{2(\eta + \epsilon_1)^2 + \sqrt{2(\eta + \epsilon_1)}^2} \\
            \leq & \sqrt{2(\eta + \epsilon_1)^2 + 2(\eta + \epsilon_1)} \\
            \leq & 2\sqrt{\eta + \epsilon_1} .
        \end{align*}
    
    \item Now let us assume \(2(\eta + \epsilon_1) < \delta^2\). First, rewrite \(\sqrt{1 - \|u\|^2}\) as
        \begin{align*}
            C = & \sqrt{1 - \|u\|^2} \\
            = & \sqrt{1 - \|v\|^2 + \|v\|^2 - \|u\|^2)} \\
            = & \sqrt{(1 - \|v\|^2)\left(1 + \frac{\|v\|^2 - \|u\|^2}{1 - \|v\|^2}\right)} \\
            = & \sqrt{1 - \|v\|^2}\sqrt{1 + \frac{\|v\|^2 - \|u\|^2}{1 - \|v\|^2}}.\\
        \end{align*}
        Now when \(\|v\|^2 - \|u\|^2 \ll 1 - \|v\|^2\), we have \(\sqrt{1 + \frac{\|v\|^2 - \|u\|^2}{1 - \|v\|^2}} \sim 1 + \frac{1}{2}\frac{\|v\|^2 - \|u\|^2}{1 - \|v\|^2}\) from a binomial approximation. In particular
        \begin{equation} \label{BinomialApproximation}
            \left|\sqrt{1 + \frac{\|v\|^2 - \|u\|^2}{1 - \|v\|^2}} - 1 - \frac{1}{2}\frac{\|v\|^2 - \|u\|^2}{1 - \|v\|^2}\right| \leq \frac{1}{4}\left(\frac{\|v\|^2 - \|u\|^2}{1 - \|v\|^2}\right)^2 .
        \end{equation}
        
        Now when we calculate the difference, we get:
        \begin{align*}
            |C - \bra{\phi}\ket{\psi}| = & \left| \sqrt{1 - \|u\|^2} - \sqrt{1 - \|v\|^2} \right| \\
            = & \left| \sqrt{1 - \|v\|^2}\sqrt{1 + \frac{\|v\|^2 - \|u\|^2}{1 - \|v\|^2}} - \sqrt{1 - \|v\|^2} \right| \\
            = & \sqrt{1 - \|v\|^2} \left|\sqrt{1 + \frac{\|v\|^2 - \|u\|^2}{1 - \|v\|^2}} - 1 \right| \\
            \leq & \sqrt{1 - \|v\|^2}\left( \left| 1 + \frac{1}{2}\frac{\|v\|^2 - \|u\|^2}{1 - \|v\|^2} - 1 \right| + \frac{1}{4}\left(\frac{\|v\|^2 - \|u\|^2}{1 - \|v\|^2}\right)^2 \right)\\
            \leq & \sqrt{1 - \|v\|^2}\left(\frac{1}{2}\frac{|\|v\|^2 - \|u\|^2|}{1 - \|v\|^2} + \frac{1}{4}\frac{\|v\|^2 - \|u\|^2}{1 - \|v\|^2}\frac{\|v\|^2 - \|u\|^2}{1 - \|v\|^2} \right)\\
            \leq & \frac{|\|v\|^2 - \|u\|^2|}{\sqrt{1 - \|v\|^2}}\left(\frac{1}{2} + \frac{1}{4}\frac{|\|v\|^2 - \|u\|^2|}{1 - \|v\|^2}\right) .
        \end{align*}
        
        Finally, recalling that that \cref{vdeltaBound}, \cref{SquareBound} and our assumption that \(2(\eta + \epsilon_1) < \delta^2\), we have
        \begin{align*}
            |C - \bra{\phi}\ket{\psi}| \leq & \frac{|\|v\|^2 - \|u\|^2|}{\sqrt{1 - \|v\|^2}}\left(\frac{1}{2} + \frac{1}{4}\frac{|\|v\|^2 - \|u\|^2|}{1 - \|v\|^2}\right) \\
            \leq & \frac{2(\eta + \epsilon_1)}{\delta}\left(\frac{1}{2} + \frac{1}{4}\frac{2(\eta + \epsilon_1)}{\delta^2}\right) \\
            \leq & \frac{2(\eta + \epsilon_1)}{\delta} .
        \end{align*}
        
        So together
        \begin{align*}
            \|\ket{\theta} - \ket{\psi}\| = & \sqrt{2(\eta + \epsilon_1)^2 + |C - \bra{\phi}\ket{\psi}|^2} \\
            \leq & \sqrt{2(\eta + \epsilon_1)^2 + \left(\frac{2\left(\eta + \epsilon_1\right)}{\delta}\right)^2} \\
            \leq & 4 \frac{\eta + \epsilon_1}{\delta}  .
        \end{align*}
        
        So \( \|\ket{\theta} - \ket{\psi}\| \leq \epsilon_2\).
        
    \end{itemize}
    
    \item[$C = 0$]
    Following a similar proof to the last case, the unnormalized version of the formula satisfies the relation. The concern is that normalization might make the difference worse. Well suppose we have two vectors, \(u, v\), with \(\|u\| > \|v\| = 1\). We would like to show \(\|\frac{u}{\|u\|} - v\| \leq \|u - v\|\). Well,
    \begin{align*}
        \|\frac{u}{\|u\|} - v\|^2 = & \sum_i \left(\frac{u_i}{\|u\|} - v_i\right)^2 \\
        = & \sum_i \frac{u_i^2}{\|u\|^2} + \sum_i v_i^2 - 2 \sum_i \frac{u_i}{\|u\|} v_i \\
        = & 2 - 2 \frac{1}{\|u\|} \sum_i u_i v_i \\
        \leq & 2\|u\| - 2 \sum_i u_i v_i \\
        \leq & \|u\|^2 + 1 - 2 \sum_i u_i v_i \\
        = & \sum_i u_i^2 + \sum_i v_i^2 - 2 \sum_i u_i v_i \\
        = & \sum_i (u_i - v_i)^2 \\
        = & \|u - v\|^2 .
    \end{align*}
\end{itemize}
That proves that when the output is ``nice'', the output is ``good'', and we know by how we estimated the angles the output is nice with probability at least $1- \gamma$. Note, this actually gives an \(\epsilon_2\) approximation of the rotation $R$ required for \cref{Orthogonalization}.

Now to compute $U$, we just use \cref{Orthogonalization} to get a circuit $A$. All that remains to be done is verifying the size and approximation is what we claimed. Well $A$ has size \(O \left(\frac{k \ln(1/\epsilon_3)}{\delta^2}\right)\), and would have error \(\epsilon_3\) if each \(U_i\) and \(R\) was a 0 approximation. Instead, each \(U_i\) is an \(\eta_i\) approximation and are each called \(O \left(\frac{\ln(1/\epsilon_3)}{\delta^2}\right)\) times while \(R\) is an \(\epsilon_2\) approximation called \(O \left(\frac{\ln(1/\epsilon_3)}{\delta^2}\right)\) times. Thus $A$ is an
\begin{align*}
    & O \left(\frac{\ln(1/\epsilon_3)}{\delta^2}\right) \left(\epsilon_2 + \sum_{i \in k}\eta_i\right) + \epsilon_3 \\
    =& O \left((\eta + \epsilon_2) \frac{\ln(1/\epsilon_3)}{\delta^2}\right) + \epsilon_3 \\
    =& O \left(\epsilon_2 \frac{\ln(1/\epsilon_3)}{\delta^2}\right) + \epsilon_3 \\
\end{align*}
approximation.

\end{proof}

Now applying \cref{ApproximateOrthogonalComponents} in the straightforward way gives a weak version of \cref{ApproximateBasis}. But we might get unlucky and constantly get small \(\delta \sim \epsilon\). This would rule out having small \(\epsilon\) and large \(k\) at the same time for a polytime algorithm. We have to be careful to construct our orthogonal components in an order that only has non constant \(\delta\) when orthogonalizing small circuits. We achieve this in two ways.

\begin{enumerate}
    \item We can actually orthogonalize with respect to parts of our basis at a time. So if $\delta$ is not constant, we orthogonalize with respect to the parts of the previously computed basis that have ``small'' circuits. Hopefully the resulting state has a large \(\delta\).
    
    \item Every time the above procedure fails to give a state with constant \(\delta\), we reset the group of circuits we consider ``small'' to be all the circuits.
    
    The key insight here is the above resetting can only occur at most \(O(\ln(k/\epsilon))\) times, after which point the initial \(\delta\) must be less than \(\epsilon\).
\end{enumerate}

The straightforward thing is to orthogonalize each \(V_i\) in order. But this may lead having to pay the \(\frac{1}{\epsilon}\) factor $k$ times. So instead, each component we orthogonalize will proceed in two steps: orthogonalize on the small circuits, orthogonalize on the rest. To do this, we will keep track of an ``orthogonalization depth'' which will intuitively count the number of times we had to pay the \(\frac{1}{\epsilon}\) factor. But specifically, it will maintain the number of times we added more circuits to our set of ``small circuits''.

Now we have all the tools to prove \cref{ApproximateBasis}.

\ApproximateBasis*
\begin{proof}
    We will use \(\epsilon_1 > 0\) as some constant that will be decided later which will be polynomial in the size of the circuit. Similarly, we have \(\gamma' > 0\) as the chance measurement at any step fails. These are the constants we will use with \cref{ApproximateOrthogonalComponents}.
    
    So let \(A\) be the constant from the last line in \cref{ApproximateOrthogonalComponents} for the case where \(2(\eta + \epsilon_1) \leq \delta^2\). That is, when the output is good, we have an
    \[A \frac{(\eta + \epsilon_1)\ln(1/\epsilon_1)}{\delta^3} \]
    approximation (where $\epsilon_3 = \epsilon_1$. Similarly, the resulting circuit for some constant $C$, makes less than
    \[C \frac{\ln(1/\epsilon_1)}{\delta^2}\]
    calls to each oracle and size less than
    \[C k \frac{\ln(1/\epsilon_1)}{\delta^2}.\]
    
    The proof is then by induction on each added basis state and circuit to approximate it.
    \begin{description}
        \item[Base Case] To start, let:
        \begin{enumerate}
            \item \(d_0 = 0\), our initial orthogonalization depth.
            \item $Y_0 = \emptyset$, the basis we know how to compute so far.
            \item \(X_0 = \emptyset\), the circuits preparing $Y_0$.
            \item \(Y'_0 = \emptyset\), the states $X'_0$ approximately produce.
            \item \(X'_0 = \emptyset\), the initial set of "small" circuits.
            \item \(Z_0 = \emptyset\), are the indices $i$ that we have found are in $\ket{\psi_i}$ are in our found basis.
        \end{enumerate}
        \item[Assume]
            Let \(j \in \{0,...,k-1\}\). Then the inductive assumption is:
            \begin{enumerate}
                \item $d_j = d$ is our current orthogonalization depth.
                \item \(Y_j\) is a set of at most $j+1$ orthonormal states, in the span of $(\ket{\psi_i})_{i \in [k]}$.
                \item \(X_j\) of circuits, where for \(\ket{a} \in Y_j\) there is a \(x_a \in X_j\) that \(\eta_{x_a}\) approximates a unitary preparing \(\ket{a}\).
                \begin{enumerate}
                    \item 
                    \[\eta = \sum_{x \in X_j} \eta_{x_a} \leq (512 A \ln(1/\epsilon_1))^j \left(8192 A\frac{k^3 \ln(1/\epsilon_1)}{\epsilon^3}\right)^d \epsilon_1 \]
                    \item \[\sum_{x \in X_j} |x| \leq \left(64 C \ln(k/\epsilon_1)\right)^j \left(\frac{256 C k^2\ln(k/\epsilon_1)}{\epsilon^2}\right)^{d}.\]
                \end{enumerate}
                \item \(Y'_j\) is a set of at most $j+1$ orthonormal states.
                \item \(X'_j\) is a set of circuits, where for \(\ket{a} \in Y'_j\) there is a \(x_a \in X'_j\) that \(\eta_{x_a}\) approximates a unitary preparing \(\ket{a}\).
                \begin{enumerate}
                    \item \[\eta' = \sum_{x \in X'_j} \eta_{x_a} \leq (512 A \ln(1/\epsilon_1))^j \left(8192 A \frac{k^3 \ln(1/\epsilon_1)}{\epsilon^3}\right)^{d-1} \epsilon_1 \]
                    \item \[\sum_{x \in X'_j} |x| \leq \left(64 C \ln(k/\epsilon_1)\right)^j \left(\frac{256 C k^2\ln(k/\epsilon_1)}{\epsilon^2}\right)^{d-1}.\]
                \end{enumerate}
                \item \(Z_j \subseteq [k]\) with \(|Z_j| = j\) and
                \begin{enumerate}
                    \item for all \(i \in Z_j\), \(\Delta(Y_j, \ket{\psi_i}) \leq \frac{\epsilon}{4k}\).
                    \item for \(i \in [k] \setminus Z_j\), we have \(\Delta(Y'_j, \ket{\psi_i})^2 \leq \left(\frac{2}{3}\right)^d\).
                \end{enumerate}
            \end{enumerate}
            
        \item[Implies]
            First, use \cref{ApproximateOrthogonalComponents} to calculate for each \(i \in [k] \setminus Z_j\) the components of \(\ket{\psi_i}\) along each state in \(Y_j\). We do this with parameters such that it will succeed for every $i$ with additive error at most \(\eta_i + \frac{\epsilon_1}{k}\) with probability at least \(1 - k \gamma'\). Suppose it succeeds (that is, our approximation is ``nice'' for each $i$).  This takes only polynomial time in $k, 1/\epsilon_1, \ln(1/\gamma')$ and the size of any circuits in $X_j$.
            
            Using these components, we can approximate the distance of $\ket{\psi_i}$ to $Y'_j$, and also the distance from $Y_j$ of the orthogonal component of $\ket{\psi_i}$ from $Y_j$. That is, we can approximate $\Delta(Y'_j, \ket{\psi_i})$ and \(\Delta(Y_j, \Phi(Y'_j, \ket{\psi_i}))\), $\Delta(Y_j, \ket{\psi_i})$ directly. We will call the function that computes these approximate orthogonal components \(\Delta'\).
            
            Suppose for some $i$
            $$\Delta'(Y_j, \ket{\psi_i}) < \frac{\epsilon}{8k}.$$
            Well \(\Delta'(Y_j, \ket{\psi_i})\) is off by at most $2 \sqrt{\eta + \epsilon_1}$, since the output is ``good'' from \cref{ApproximateOrthogonalComponents}. We will choose \(\epsilon_1\) so that this is less than \( \frac{\epsilon}{16k}\). Then
            $$\Delta(Y_j, \ket{\psi_i}) < \frac{\epsilon}{4k}.$$
            Then let
            \begin{enumerate}
                \item $d_{j+1} = d_j$.
                \item $Y_{j+1} = Y_j$.
                \item $X_{j+1} = X_j$.
                \item $Y'_{j+1} = Y'_j$.
                \item $X'_{j+1} = X'_j$.
                \item \(Z_{j+1} = Z_j \cup \{i\}\).
            \end{enumerate}
            Then all the assumptions still hold, and we can move onto the next $j$.
            
            Otherwise, see that we have 
            $$\delta'_i = \Delta(Y'_j, \ket{\psi_i}) \geq \Delta(Y_j, \ket{\psi_i}) \geq \frac{\epsilon}{16k}$$
            by a similar argument.

            Now, for each \(i \in [k] \setminus Z_j\) calculate a numerical approximation of:
            \begin{align*}
                \delta^*_i \sim & \Delta(Y_j, \Phi(Y'_j, \ket{\psi_i})) \\
                = & \frac{\Delta(Y_j,  \ket{\psi_i})}{\Delta(Y'_j,  \ket{\psi_i})}.
            \end{align*}
            
            Suppose $2\sqrt{ (\eta + \epsilon_1)} < \frac{\epsilon^2}{1000 k^2}$. We will later choose \(\epsilon_1\) so this is true. Then the error of \(\delta'_i\) will be at most the error of \(\Delta(Y_j, \Phi(Y'_j, \ket{\psi_i}))\) (at most \(2\sqrt{ (\eta + \epsilon_1)}\)), divided by \(\Delta(Y'_j,  \ket{\psi_i}) \) (at least \(\frac{\epsilon}{16k}\)). So the error of $\delta'_i$ is at most \( \frac{\epsilon}{16k} < \frac{1}{4}\).
            
            Either for some $i$, \(\delta^*_i > \frac{1}{2}\), or for all it is not.
            \begin{itemize}
                \item $\exists i \in [k] \setminus Z_j: \delta^*_i > \frac{1}{2}$: then take such i. Then
                \[\delta_i = \Delta(Y_j, \Phi(Y'_j, \ket{\psi_i})) > \frac{1}{4}.\]
                
                Now we want to preprocess $V_i$ to get a $V'_i$ preparing $\ket{\psi_i'}$ that is orthogonal to $Y'_j$. If we have \(Y'_j = \emptyset\) and $d=0$, then just let \(V'_i = V_i\) and \(\ket{\psi'_i} = \ket{\psi_i}\) and move onto the next step.
                
                Otherwise, run \cref{ApproximateOrthogonalComponents} using \(V_i\) and \(X'_j\) to get circuit \(V'_i\) that approximately computes \(\ket{\psi'_i} = \Phi(Y'_j, \ket{\psi})\). Notice that we already assumed the output is ``nice'' since we have already approximated these angles with high enough accuracy.
                
                Now we already have that $\delta'_i > \epsilon/16k$. Then for sufficiently small $\epsilon_1$, we have \(V'_i\) is a circuit $\eta'_i$-approximately preparing \(\ket{\psi'_i}\) where
                \begin{align*}
                    \eta_i' \leq & A \frac{(\eta' + \epsilon_1)\ln(1/\epsilon_1)}{\delta'^3_i} \\
                    \leq & A 8192 k^3  \frac{\ln(1/\epsilon_1)}{\epsilon^3} \eta' \\
                    \leq & A 8192 k^3  \frac{\ln(1/\epsilon_1)}{\epsilon^3}  (512 A \ln(1/\epsilon_1))^j \left(8192 A \frac{k^3 \ln(1/\epsilon_1)}{\epsilon^3}\right)^{d-1} \epsilon_1 \\
                    \leq & (512 A \ln(1/\epsilon_1))^j \left(8192 A \frac{k^3 \ln(1/\epsilon_1)}{\epsilon^3}\right)^{d} \epsilon_1 .
                \end{align*}
                Further, \(V'_i\) has size
                \begin{align*}
                    |V'_i| \leq &C \frac{\ln(k/\epsilon_1)}{\delta'^2_i} \sum_{x \in X'_j} |x| \\
                    \leq & C 256 k^2 \frac{\ln(k/\epsilon_1)}{\epsilon^2} \sum_{x \in X'_j} |x| \\
                    \leq & C 256 k^2 \frac{\ln(k/\epsilon_1)}{\epsilon^2} \left(64 C \ln(k/\epsilon_1)\right)^j \left(\frac{256 C k^2\ln(k/\epsilon_1)}{\epsilon^2}\right)^{d-1} \\
                    \leq & \left(64 C \ln(k/\epsilon_1)\right)^j \left(\frac{256 C k^2\ln(k/\epsilon_1)}{\epsilon^2}\right)^{d} .
                \end{align*}
                
                Now given \(V'_i\) approximating \(\ket{\psi_i'}\), we run \cref{ApproximateOrthogonalComponents} again to get circuit \(U_i\) which approximates \(\ket{\phi_i} = \Phi(Y_j, \ket{\psi_i'})\). This only takes a polynomial sized circuit and fails with probability at most $\gamma'/k$.
                
                Then the error on \(U_i\), \(\eta_i\), is
                \begin{align*}
                    \eta_i \leq & A \frac{(\eta + \eta_i' + \epsilon_1)\ln(1/\epsilon_1)}{\delta_i^3} \\
                    \leq & A 128 (\eta + \eta'_i)\ln(1/\epsilon_1) \\
                    \leq & A 256 \ln(1/\epsilon_1) (512 A \ln(1/\epsilon_1))^j \left(8192 A \frac{k^3 \ln(1/\epsilon_1)}{\epsilon^3}\right)^{d} \epsilon_1 \\
                    \leq & \frac{1}{2}(512 A \ln(1/\epsilon_1))^{j+1} \left(8192 A \frac{k^3 \ln(1/\epsilon_1)}{\epsilon^3}\right)^{d} \epsilon_1.
                \end{align*}
                And for the size of $U_i$, 
                \begin{align*}
                    |U_i| \leq & C \frac{\ln(k/\epsilon_1)}{\delta^2} (|V'_i| + \sum_{x \in X_j} |x|)\\
                     \leq & C 16 \ln(k/\epsilon_1) 2 \left(64 C \ln(k/\epsilon_1)\right)^j \left(\frac{256 C k^2\ln(k/\epsilon_1)}{\epsilon^2}\right)^{d}\\
                     \leq & \frac{1}{2}\left(64 C \ln(k/\epsilon_1)\right)^{j+1} \left(\frac{256 C k^2\ln(k/\epsilon_1)}{\epsilon^2}\right)^{d}.
                \end{align*}
                
                Finally, let
                \begin{enumerate}
                    \item \(d_{j+1} = d_j\).
                    \item \(Y_{j+1} = Y_j \cup \{\ket{\phi_i}\}\).
                    \item \(X_{j+1} = X_j \cup \{U_i\}\). One can easily verify that
                    \begin{enumerate}
                        \item \[\sum_{x \in X_{j+1}} \eta_{x_a} \leq (512 A \ln(1/\epsilon_1))^{j+1} \left(8192 A \frac{k^2\ln(1/\epsilon_1)}{\epsilon_1^3}\right)^{d} \epsilon_1. \]
                        \item \[\sum_{x \in X_{j+1}} |x| \leq \left(64 C \ln(k/\epsilon_1)\right)^{j+1} \left(\frac{256 C k^2\ln(k/\epsilon_1)}{\epsilon^2}\right)^{d}.\]
                    \end{enumerate}
                    \item \(Y'_{j+1} = Y'_j\)
                    \item \(X'_{j+1} = X'_j\). Size and error constraints still holds since no change.
                    \item \(Z_{j+1} = Z_j \cup \{i\}\).
                    \begin{enumerate}
                        \item For $i \in \Z_{j+1}$, we only added one $i$ to $Z_{j+1}$ and $\ket{\phi_i}$ to $Y_{j+1}$ so that 
                        \[\Delta(Y_{j+1}, \ket{\psi_i}) = 0\]
                        and the rest did not increase.
                        \item The case where $i \notin Z_{j+1}$ still holds since $d$ and $Y'_J$ is unchanged.
                    \end{enumerate}
                \end{enumerate}

                \item $\forall i \in [k] \setminus Z_j: \delta^*_i \leq \frac{1}{2}$. Then in this case, orthogonalizing with respect to \(Y'_j\) first won't help as much. So just choose some arbitrary \(i \in [k] \setminus Z_j\). \(Y'_{j+1} = Y_j\), \(X'_{j+1} = X_j\), \(Z_{j+1} = Z_j \cup \{i\}\), and \(d_{j+1} = d + 1\). Let
                \[\delta_i = \Delta(Y_j, \ket{\psi_i}) \geq \frac{\epsilon}{16k}.\]
                
                Use \cref{ApproximateOrthogonalComponents} with \(V_i\) and \(X_j\) to get \(U_i\) approximating \(\ket{\phi_i} = \Phi(Y_j, \ket{\psi_i})\). The error of \(U_i\) is
                \begin{align*}
                    \eta_i \leq & A \frac{(\eta + \epsilon_1)\ln(1/\epsilon_1)}{\delta_i^3} \\
                    \leq & A 8192 k^3  \frac{\ln(1/\epsilon_1)}{\epsilon^3} \eta \\
                    \leq & A 8192 k^3  \frac{\ln(1/\epsilon_1)}{\epsilon^3} (512 A \ln(1/\epsilon_1))^j \left(8192 A\frac{k^3\ln(1/\epsilon_1)}{\epsilon^3}\right)^d \epsilon_1 \\
                    \leq & (512 A \ln(1/\epsilon_1))^j \left(8192 A\frac{k^3\ln(1/\epsilon_1)}{\epsilon^3}\right)^{d+1} \epsilon_1 .
                \end{align*}
                and the size of $U_i$ is
                \begin{align*}
                    |U_i| \leq & C \frac{\ln(k/\epsilon_1)}{\delta^2} \sum_{x \in X_j} |x| \\
                    \leq & C 256 k^2 \frac{\ln(k/\epsilon_1)}{\epsilon^2} \sum_{x \in X_j} |x| \\
                    \leq & C 256 k^2 \frac{\ln(k/\epsilon_1)}{\epsilon^2} \left(64 C \ln(k/\epsilon_1)\right)^j \left(\frac{256 C k^2\ln(k/\epsilon_1)}{\epsilon^2}\right)^d \\
                    \leq & \left(64 C \ln(k/\epsilon_1)\right)^j \left(\frac{256 C k^2\ln(k/\epsilon_1)}{\epsilon^2}\right)^{d+1} .
                \end{align*}
                
                Then let
                \begin{enumerate}
                    \item \(d_{j+1} = d + 1\).
                    \item \(Y_{j+1} = Y_j \cup \{\ket{\phi_i}\).
                    \item \(X_{j+1} = X_j \cup \{U_i\}\), a quick calculation shows the size and approximation assumptions are satisfied.
                    \item \(Y'_{j+1} = Y_j\).
                    \item \(X'_{j+1} = X_j\), already satsifies size and approximation assumptions since $X_j$ satisfied its.
                    \item \(Z_{j+1} = Z_j \cup \{i\}\).
                    \begin{enumerate}
                        \item For $m \in Z_{j+1}$, we only added one $i$ which exactly has \(\Delta(Y_{j+1}, \ket{\psi_m}) = 0\). The rest have not increased.
                        \item Choose an \(m \in [k] \setminus Z_{j+1}\). We want to show 
                        \[\Delta(Y'_{j+1}, \ket{\psi_m})^2 \leq \left(\frac{2}{3}\right)^{d+1}.\]
                
                        By assumption of this case, \(\delta^*_m \leq \frac{1}{2}\). Because the error is less than \( \frac{\epsilon}{16k} < \frac{1}{16}\), we have
                        \begin{align*}
                            \frac{\Delta(Y_j, \ket{\psi_m})}{\Delta(Y'_j, \ket{\psi_m})} \leq & \frac{2}{3} \\
                            \Delta(Y_j, \ket{\psi_m})\leq & \frac{2}{3} \Delta(Y'_j, \ket{\psi_m}) .
                        \end{align*}
                        Since $Z_{j} \subseteq Z_{j+1}$, $m \in [k] \setminus Z_j$. So by inductive assumption \[\Delta(Y'_j, \ket{\psi_m}) \leq \left(\frac{2}{3}\right)^d.\]
                        Then
                        \begin{align*}
                            \Delta(Y'_{j+1}, \ket{\psi_m}) = & \Delta(Y_j, \ket{\psi_m}) \\
                            \leq & \frac{2}{3}\Delta(Y'_j, \ket{\psi_m}) \\
                             \leq & \frac{2}{3}\left(\frac{2}{3}\right)^d \\
                             \leq & \left(\frac{2}{3}\right)^{d_{j+1}}.
                        \end{align*}
                    \end{enumerate}
                \end{enumerate}
                
            \end{itemize}
            
            Thus our inductive hypothesis holds.
            
    \end{description}
    We can bound the final depth, $d$, by the fact that for $d$ to increase, we must not have any \(i \in [k] \setminus Z_j\) with
    \[\Delta(Y'_j, \ket{\psi_i}) \leq \frac{\epsilon}{16k},\]
    otherwise, we would have added that to \(Y\) first. Thus
    \begin{align*}
        \left(\frac{2}{3}\right)^{d-1} \geq & \frac{\epsilon}{16k} \\
        d \leq & O(\ln(k/\epsilon)).
    \end{align*}
    We also trivially have that \(d \leq k\). Thus
    \[d \leq O(\min\{k, \ln(k/\epsilon)\}).\]
    
    At the end of this induction, we will have circuits \(X\) giving us the required basis with cumulative error
    \begin{align*}
        \eta = & \sum_{x \in X} \eta_{x} \\
        \leq & (512 A \ln(1/\epsilon_1))^k \left(8192 A\frac{k^3 \ln(1/\epsilon_1)}{\epsilon^3}\right)^{O(\min\{k, \ln(k/\epsilon)\})} \epsilon_1 \\
        \leq & \poly\left(\ln(1/\epsilon_1)^k \left(\frac{k \ln(1/\epsilon_1)}{\epsilon}\right)^{\min\{k, \ln(k/\epsilon)\}}\right) \epsilon_1
    \end{align*}
    and size
    \begin{align*}
        \sum_{x \in X} |x| \leq & \left(64 C \ln(k/\epsilon_1)\right)^k \left(\frac{256 C k^2\ln(k/\epsilon_1)}{\epsilon^2}\right)^{O(\min\{k, \ln(k/\epsilon)\})} \\
        \leq & \poly\left(\left(\ln(k/\epsilon_1)\right)^k \left(\frac{k \ln(k/\epsilon_1)}{\epsilon^2}\right)^{\min\{k, \ln(k/\epsilon)\}}\right).
    \end{align*}
    
    This is the claimed size. Now we want to make sure this final error has \(\eta < \frac{\epsilon}{16k}\), and $\sqrt{3 (\eta + \epsilon_1)} < \frac{\epsilon^2}{1000 k^2}$, and inside of \cref{ApproximateOrthogonalComponents}, I also need \(2(\eta + \epsilon_1) \leq \left(\frac{\epsilon}{16k}\right)^2\). This holds for some 
    \[\epsilon_1 =e^{-\poly(k + \ln(1/\epsilon))}.\]
    Take \(\epsilon_1\) to be such that it is still polynomial in the output circuit size.
    
    Then the size is
    \begin{align*}
        \sum_{x \in X} |x| \leq & \poly\left(\left(k + \ln(1/\epsilon)\right)^k \left(\frac{k}{\epsilon}\right)^{\min\{k, \ln(k/\epsilon)\}}\right).
    \end{align*}
    Now to further simplify this, let us compare $k$ to $\ln(1/\epsilon)$.
    \begin{description}
        \item[$k \leq \ln(1/\epsilon)$]: Then
        \begin{align*}
            \left(k + \ln(1/\epsilon)\right)^k \left(\frac{k}{\epsilon}\right)^{\min\{k, \ln(k/\epsilon)\}} \leq &
            \left(2\ln(1/\epsilon)\right)^{k} \left(\frac{\ln(1/\epsilon)}{\epsilon}\right)^{k} \\
            \leq & \poly\left( \left(\ln(1/\epsilon)\right)^{k} \left(\frac{1}{\epsilon}\right)^{k} \right) \\
            \leq & \poly\left(\left(\frac{1}{\epsilon}\right)^{k} \right) \\
            \leq & \poly\left(k^k \left(\frac{1}{\epsilon}\right)^{\min\{k, \ln(1/\epsilon)\}} \right).
        \end{align*}
        \item[$k > \ln(1/\epsilon)$]: Then
        \begin{align*}
            \left(k + \ln(1/\epsilon)\right)^k \left(\frac{k}{\epsilon}\right)^{\min\{k, \ln(k/\epsilon)\}}
            \leq & \left(2k\right)^k \left(\frac{k}{\epsilon}\right)^{\min\{k, \ln(k/\epsilon)\}} \\
            \leq & \poly \left(k^k \left(\frac{k}{\epsilon}\right)^{\ln(k/\epsilon)} \right) \\
            \leq & \poly \left(2^{\ln(k) k + \ln(k/\epsilon)^2} \right) \\
            \leq & \poly \left(2^{\ln(k) k + \ln(k)^2 + \ln(k)\ln(1/\epsilon) + \ln(1/\epsilon)^2} \right) \\
            \leq & \poly \left(2^{\ln(k) k + \ln(1/\epsilon)^2} \right) \\
            \leq & \poly \left(k^k \left(\frac{1}{\epsilon}\right)^{\ln(1/\epsilon)} \right) \\
            \leq & \poly \left(k^k \left(\frac{1}{\epsilon}\right)^{\min\{k, \ln(1/\epsilon)\}} \right).
        \end{align*}
    \end{description}
    In either case, 
    \begin{align*}
        \sum_{x \in X} |x| \leq & \poly \left(k^k \left(\frac{1}{\epsilon}\right)^{\min\{k, \ln(1/\epsilon)\}} \right).
    \end{align*}

    Our angle approximations were all computed with polynomial sized circuits, and the total risk of failure is at most \(\gamma'\) times a polynomial in the size of our output circuits. The size of each angle approximation circuit is of size polynomial in \(1/\epsilon_1\) and k, which is also polynomial in the size of the output circuit,  times \(\ln(1/\gamma')\). Then the probability of any approximation failing is at most \(\gamma'\) times the size of the circuit, but by taking \(\gamma' = \gamma/(\sum_{x \in X} |x| )\), the circuit producing circuit still has polynomial size in the circuit output size times \(\ln(1/\gamma)\).
\end{proof}

\subsection{Proving Unitarization}

And using \cref{DefferedRotation} and \cref{ApproximateBasis} gives our main theorem, \cref{CleanUnitary}.

\CleanUnitary*

\begin{proof}
    Let $\epsilon_1 = \frac{\epsilon}{100000 k^3}$. Run \cref{ApproximateBasis} to get with probability $1 - \epsilon/2$ an $\epsilon_1$ approximate orthonormal basis, $Y = (\ket{\phi_i})_{i \in [m]}$, prepared by $X$, for  $(\ket{i-1})_{i \in [k]} \cup (\ket{\psi_i})_{i \in [k]}$.
    
    Then run \cref{RobustAngleApproximation} to, with probability $1 - \epsilon/2$, approximate the components of each $\ket{\psi_i}$ along each state in $Y$. Call the estimate of $\bra{\phi_j}\ket{\phi_i}$ as $v^j_i$. We run \cref{RobustAngleApproximation} with parameters so that the error on $v^j_i$ is at most $\eta_j + \epsilon_1/(2k)$ (where $\eta_j$ is the error of the circuit producing $\ket{\phi_j}$). Then the total error of the approximation for $\ket{\psi_i}$ using $v_i$ is at most $\epsilon_1 + \epsilon_1 + \epsilon_1 = 3 \epsilon_1 $ from the error on $\eta_x$, the error on the approximation, and the component of $\ket{\psi_i}$ outside $Y$.
    
    Doing this for all $\ket{\psi_i}$ with probability $1 - \epsilon/2$ gives us an $3 \epsilon_1$ approximation of each $\ket{\psi_i}$. Then after applying \cref{ExtrapolateUnitary} with our $v$ as the first columns, we get the coefficients $v'$ we need for \cref{DefferedRotation}. For $i \in [k]$, using $v'_i$ to approximate $\ket{\psi_i}$ gives a state, $\ket{\psi'_i}$, which is only $195 k \epsilon_1$ off from $\ket{\psi_i}$. Then applying \cref{DefferedRotation} with $v'$ and $X$ give an $O(\epsilon_1)$ (for some constant less than 100) approximation of the unitary preparing $\ket{\psi'_i}$ from $\ket{i-1}$ and changing no state orthogonal to $Y$.
    
    The resulting circuit $\epsilon_2 = 300 k \epsilon_1$ approximates a unitary, $U$, acting only on the subspace spanned by $Y$ such that $U\ket{i-1}$ gets mapped to a $\ket{\psi_i'}$ which is within $ \epsilon_2$ of $\ket{\psi_i}$. This is almost good enough, but we want to show we approximate the unitary that exactly takes $\ket{i-1}$ to $\ket{\psi_i}$. But this is actually equivalent, within a factor of $k$.
    
    Let $Y'$ be the span of $Y$, $X'$ be the span of $(\ket{i-1})_{i \in [k]}$, and $Z'$ be the space spanned by $(\ket{i-1})_{i \in [k]} \cup (\ket{\psi_i})_{i \in [k]}$. Then by choice of $Y$, $Y' \subseteq Z$.
    
    What we have is an $\epsilon_2$ approximation of a unitary, $U$, acting only on $Y'$, taking $\ket{i-1}$ to $\ket{\psi'_i}$, which is within $\epsilon_2$ of $\ket{\psi_i}$. What we need is an approximation of a unitary acting only in $Z$ that exactly takes $\ket{i-1}$ to $\ket{\psi_i}$, $U'$. I construct such a $U'$ and show that $U$ already approximates it.
    
    We describe $U$ and $U'$ by describing what they do on some basis for $Z$. Let $V_X$ be the computational basis for $X'$, $V_Y$ some basis for $Y' \setminus X' $, and $V_Z$ be some basis for $Z' \setminus Y'$. Then we know $U$ takes $V_X$ to $V'_X = (\ket{\psi'_i})_{i \in [k]}$, $U$ takes $V_Y$ to some orthonormal $V'_Y$ in $Y'$, and $U$ is identity on $V_Z$. Let $V = V'_X \cup V'_Y \cup V_Z$. If we can find a basis containing $(\ket{\psi'_i})_{i \in [k]}$ and approximating $V'_Y$ and $V_Z$, we are done.
    
    Luckily, we have \cref{GramSchmidtIsStab}. So let $W_X = (\ket{\psi_i})_{i \in [k]}$. Then run Gram-Schmidt on $W' = W_X \cup V'_Y \cup V_Z$, in that order to get basis $W$. Each element of $W'$ $\epsilon_2$ approximates the orthonormal vectors, $V$. Thus by \cref{GramSchmidtIsStab}, $W$ $65(2k) \epsilon_2$ approximates the orthonormal vectors $V$. Thus the unitary $U'$ that maps $V_X \cup V_Y \cup V_Z$ to $W$ instead of $V$ is only $130k\epsilon_2$ away from $U$.
    
    That is, $U$ $130k\epsilon_2$ approximates $U'$. And $U'$ exactly computes $\ket{\psi_i}$ from $\ket{i-1}$, and is identity on anything orthogonal to $Z$. Then the circuit
    \begin{align*}
        300 k \epsilon_1 + 130k\epsilon_2 \leq & 300 k \epsilon_1 + 40000 k^2\epsilon_1 \\
        \leq & \epsilon
    \end{align*}
    approximates $U'$.
    
    The most expensive step in this procedure is \cref{ApproximateBasis} , which takes time
    \begin{align*}
        \poly(k^k \left(\frac{1}{\epsilon_1}\right)^{\min\{k, \ln(1/\epsilon_1)\}} = & \poly(k^k \left(\frac{k}{\epsilon}\right)^{\min\{k, \ln(k/\epsilon)\}} \\
        = & \poly(k^k \left(\frac{1}{\epsilon}\right)^{\min\{k, \ln(1/\epsilon)\}}.
    \end{align*}
    With the last line following from the same sort of argument used at the end of \cref{ApproximateBasis}.
    
    Only two steps could have failed, \cref{ApproximateBasis} and \cref{ExtrapolateUnitary}, both fail with probability at most $\epsilon/2$, so we succeed with probability at least $1 - \epsilon$.
\end{proof}

\section{Future Work}

This algorithm gives a way to synthesize non unitary quantum circuits into a combined unitary operation, but has several limitations. First, it is not an exact algorithm. Doing this exactly seems hard as previously described, but we could not rule it out. Is there an exact solution? Perhaps even one using a constant number of queries per oracle?

It is relatively slow, requiring running time polynomial in \(k^k\). Is this fundamental? Can a better algorithm improve this to polynomial dependence on $k$? Or even just $2^{O(k)}$? Further, we only have polynomial time in \(1/\epsilon\) for constant $k$. Can we get polynomial dependence on $1/\epsilon$ for larger $k$?

One approach to better dependence on $\epsilon$ is a more fine grain approach to our different stages of orthogonalization. As you may recall, in our approach, we kept a set of cheaper, older orthogonal components, $A$, and a newer set $B$. We would try to find whichever new state was farthest from $B$ after removing the component in $A$, and add it. If that component was still too close, we would add $B$ to $A$ and continue. The insight is that after we do this a few times, all remaining states would need to be close enough to $A$ or $B$ to give an approximate basis.

Intuitively, the overhead from orthogonalizing $B$ is inversely proportional to the distance of the remaining states from $B$. And the closer $B$ is to the remaining the states, the fewer times we have to add $B$ to $A$. One approach could be to have many levels instead of two, and just add a new level each time we would normally add $B$ to $A$. This seems like a reasonable solution, but is much more complicated to analyze.

Another part of the algorithm that could possibly be improved is the orthogonal component algorithm itself. The algorithm feels to us almost like a classical algorithm. Perhaps there is some way to get some sort of a Grover speedup? This could improve the constants in the algorithm.

More importantly, we solved this problem largely because it seemed interesting on a conceptual level, but we have no application for it. Can this be used to solve any useful problem? This kind of algorithm may be useful for more linear algebra type problems when working with complex states given by quantum oracles, or when working with unitaries that act on low dimensional subspaces.

\section*{Acknowledgement}
This research was supported by NSF grant number 1705028. Thanks to Scott Aaronson, who accidentally gave me on this problem when I misunderstood a homework assignment from him.

\printbibliography

\appendix

\section{Proof of \ref{GramSchmidtIsStab}}
\GramSchmidtIsStab*
\begin{proof}
   For $i \in [k]$, define as $v_i^1 = v_i$. Then for $1 < j < i \leq k$, define
    \begin{align*}
        w_i^{j+1} = & w_i^j - \frac{\langle w_j^j, w_i^j \rangle}{\|w_j^j\|^2} w_j^j.
    \end{align*}
    Then Gram Schmidt just produces the set of vectors where 
    \[u_i = \frac{w_i^i}{\|w_i^i\|}.\]
    
    For each $1 \leq j \leq i \leq k$, for some unit vector $\phi_i^j$ in the span of $(w_i)_{i \in [j-1]}$ and some constant $\alpha_i^j > 0$
    \[w_i^j = w_i + \alpha_i^j \phi_i^j.\]
    Then
    \begin{align*}
        \langle w_j^j, w_i^j \rangle = & \langle w_j^j, w_i + \alpha_i^j \phi_i^j \rangle \\
        = & \langle w_j^j, w_i \rangle + \alpha_i^j \langle w_j^j, \phi_i^j \rangle \\
        = & \langle w_j + \alpha_j^j \phi_j^j , w_i \rangle \\
        = & \langle w_j, w_i \rangle + \alpha_j^j  \langle \phi_j^j , w_i \rangle.
    \end{align*}
    See that $\langle w_j^j, \phi_i^j \rangle = 0$, since $w_j^j$ is orthogonal to the span of $(w_i)_{i \in [j-1]}$. Since $w_j$ is within $\epsilon$ of $v_j$, $w_i$ is within $\epsilon$ of $v_i$, and $v_i$ is orthogonal to $v_j$,
    \begin{align*}
        \|\langle w_j, w_i \rangle\| \leq & 2\epsilon .
    \end{align*}
    Now for some \((\alpha_l)_{l \in [j-1]}\), 
    \[\ket{\phi_j^j} = \sum_{l \in [j-1]} \alpha_l w_l.\]
    Now we wish to bound the total size of the alphas. Then that for any $h \in [j-1]$, the following series of equations hold:
    \begin{align*}
        \langle v_h, \phi \rangle = & \sum_{l \in [j-1]} \alpha_l \langle v_h, w_l \rangle \\
        1 \geq & \|\langle v_h, \phi \rangle\|  \\
        \geq & \|\alpha_h\| - \epsilon - \sum_{l \in [j-1] \setminus \{l\}} \|\alpha_l\| 2 \epsilon \\
        j-1 \geq & \sum_{l \in [j-1]} \|\alpha_l\| - (j-2)\epsilon(1 +  \|\alpha_l\| 2) \\
        \geq & \sum_{l \in [j-1]} \|\alpha_l\| - \frac{1}{64}(1 +  \|\alpha_l\| 2) \\
        j \geq & \sum_{l \in [j-1]} \frac{31}{32}\|\alpha_l\| \\
        \|\langle v_h, \phi \rangle\| \geq & \|\alpha_h\| - \epsilon - \frac{32}{31} j 2 \epsilon \\
        \geq & \|\alpha_h\| - \epsilon - \frac{1}{31k} \\
        \geq & \|\alpha_h\| - \frac{\|\alpha_h\|}{16k} \\
        \|\langle v_h, \phi \rangle\|^2 \geq & \|\alpha_h\|^2 - \frac{1}{8k} \\
        1 = & \|\phi\| \\
        \geq & \sqrt{\sum_{l \in [j-1]} \|\langle v_h, \phi \rangle\|^2} \\
        \geq & \sqrt{\sum_{l \in [j-1]} \|\alpha_l\|^2 - \frac{\|\alpha_h\|}{8k}} \\
        \geq & \sqrt{ - \frac{1}{8} + \sum_{l \in [j-1]} \|\alpha_h\|^2} \\
        \frac{3}{\sqrt{8}} \geq & \sqrt{\sum_{l \in [j-1]} \|\alpha_h\|^2} .
    \end{align*}
    
    Then we can observe that
    
    \begin{align*}
        \|\langle \phi_j^j, w_i \rangle\| \leq &  \sum_{l \in [j-1]} \|\alpha_l\| \langle w_i, w_j \rangle \\
        \leq &  2 \epsilon \sqrt{k} \frac{3}{\sqrt{8}} \\
        \leq &  3 \epsilon \sqrt{k} .
    \end{align*}
    
    Now by induction, we will be able to show $\|w_i^j\| \geq 1/2$. This is clearly true for $j = 1$. Let
    \begin{align*}
        \gamma_i^j = & \frac{\langle w_j^j, w_i^j \rangle}{\|w_j^j\|^2} \\
         = & \frac{\langle w_j, w_i \rangle + \alpha_j^j  \langle \phi_j^j , w_i \rangle}{\|w_j^j\|^2} \\
         \|\gamma_i^j\| \leq & 4(2 \epsilon + \alpha_j^j  3 \epsilon \sqrt{k}) \\
         \leq & 12 \epsilon(1 + \alpha_j^j\sqrt{k}).
    \end{align*}
    Then
    \begin{align*}
        w_i^{j+1} = & w_i^j - \gamma_i^j w_j^j \\
        = & w_i + \alpha_i^j \ket{\phi_i^j} - \gamma_i^j(w_j + \alpha_j^j \ket{\phi_j^j}).
    \end{align*}
    From this we can conclude
    \begin{align*}
        \alpha_i^{j+1} \leq &\alpha_i^j + \|\gamma_i^j\|(1 + \alpha_j^j) \\
        \leq & \alpha_i^j + 12 \epsilon(1 + \alpha_j^j\sqrt{k})(1 + \alpha_j^j) .
    \end{align*}
    Then we will also have by induction each $\alpha_i^j \leq 32 j \epsilon$. At the same time, we will do induction on $\|w_i^j\|$.
    \begin{align*}
        \| w_i^j \| \geq & 1 - \alpha_i^j \\
        \geq 1/2.
    \end{align*}
    
    Then we can conclude for each $i$, $\alpha_i^i \leq 32 k \epsilon$. After normalizing $w_i^i$, the error is at most
    \begin{align*}
          \|w_i - u_i \| \leq & 2 \frac{\alpha_i^i}{1 - \alpha_i^i} \\
          \leq & 64 k \epsilon, \\
          \|u_i - v_i \| \leq & \|w_i - u_i \|  + \|v_i - u_i \| \\
          \leq & (64 k + 1) \epsilon .
    \end{align*}
    
\end{proof}

\end{document}